\documentclass[twocolumn,aps,prl]{revtex4-2}
 
\usepackage[utf8]{inputenc}
\usepackage{amsthm,amsmath,amssymb,latexsym,verbatim,enumerate,graphicx}
\usepackage{hyperref}
\usepackage{color}
\usepackage{dsfont}
\usepackage{subeqnarray}
\usepackage{times,txfonts}
\usepackage{lettrine}

\usepackage{tikz}

\usepackage[shortlabels]{enumitem}
\setlist{noitemsep}
\usepackage{multirow}
\usepackage{lettrine}
\usepackage{xcolor}
\usepackage{booktabs}

\usepackage{nicefrac}
\usepackage{ragged2e}
\usepackage{enumitem}
\usepackage{blindtext}

\DeclareMathAlphabet{\mathcal}{OMS}{cmsy}{m}{n} 

\definecolor{darkblue}{rgb}{0.0, 0.0, 0.5}
\definecolor{midnightblue}{rgb}{0.1378,0.3784,0.4838}
\definecolor{grayblue}{rgb}{0.3235, 0.3784, 0.4333}
\definecolor{graypink}{rgb}{0.75, 0.35, 0.25}
\definecolor{darkgreen}{rgb}{0.0, 0.25, 0.0}
\definecolor{darkpurple}{rgb}{0.3, 0.0, 0.25}
\definecolor{darkorange}{rgb}{1.0, 0.65, 0.0}
\definecolor{darkred}{rgb}{0.5, 0.0, 0.0}
\definecolor{verylightgray}{rgb}{0.9, 0.9, 0.9}
\definecolor{gray75}{gray}{0.3}

\hypersetup{colorlinks=true, citecolor=midnightblue, linkcolor=midnightblue, urlcolor=midnightblue}

\definecolor{orcidlogocol}{HTML}{A6CE39}
\tikzset{
  orcidlogo/.pic={
    \fill[orcidlogocol] svg{M256,128c0,70.7-57.3,128-128,128C57.3,256,0,198.7,0,128C0,57.3,57.3,0,128,0C198.7,0,256,57.3,256,128z};
    \fill[white] svg{M86.3,186.2H70.9V79.1h15.4v48.4V186.2z}
                 svg{M108.9,79.1h41.6c39.6,0,57,28.3,57,53.6c0,27.5-21.5,53.6-56.8,53.6h-41.8V79.1z M124.3,172.4h24.5c34.9,0,42.9-26.5,42.9-39.7c0-21.5-13.7-39.7-43.7-39.7h-23.7V172.4z}
                 svg{M88.7,56.8c0,5.5-4.5,10.1-10.1,10.1c-5.6,0-10.1-4.6-10.1-10.1c0-5.6,4.5-10.1,10.1-10.1C84.2,46.7,88.7,51.3,88.7,56.8z};
  }
}

\newcommand\orcidicon[1]{\href{https://orcid.org/#1}{\mbox{\scalerel*{
\begin{tikzpicture}[yscale=-1,transform shape]
\pic{orcidlogo};
\end{tikzpicture}
}{|}}}}

\newcommand{\dd}{\mathrm{d}}		                     									% Cool "d"

\newcommand{\hsp}{\hspace{0.05cm}}															% Just a shortcut to a spacing

\newcommand{\ii}{\mathrm{i}}
\newcommand{\ee}{\mathrm{e}}
\newcommand{\TT}{\Gamma}
\newcommand{\bra}[1]{\langle#1|}
\newcommand{\ket}[1]{|#1\rangle}

\newcommand{\tr}{\mathrm{Tr}\hspace{0.05cm}}

\newcommand{\geqcurved}{\succcurlyeq}
\newcommand{\leqcurved}{\preccurlyeq}

\theoremstyle{plain}																		% Theorem, Lemma, Proposition and Corollary

\newtheorem{thm}{Theorem}
\newtheorem{lem}[thm]{\bf Lemma}

\newtheorem{obsapp}{\bf Observation}

\newtheorem{example}{Example}

\newtheorem{observation}{\bf Observation}

\theoremstyle{remark}																		% Remark and note

\newtheorem*{rem}{Remark}

\begin{document}

\title{Two-point measurement correlations beyond the quantum regression theorem}

\author{Leonardo S. V. Santos}
\email{leonardo.svsantos@student.uni-siegen.de}
\author{Otfried G\"{u}hne}
\email{otfried.guehne@uni-siegen.de}
\author{Stefan Nimmrichter}
\email{stefan.nimmrichter@uni-siegen.de}

\date{\today}

\affiliation{Naturwissenschaftlich-Technische Fakultät, Universität Siegen, Walter-Flex-Straße 3, 57068 Siegen, Germany}

\begin{abstract}
Temporal correlations are fundamental in quantum physics, yet their computation is often challenging. The regression theorem (or hypothesis) serves as a key tool in this context, offering a seemingly straightforward approach. 
However, it fails for systems strongly coupled to their surroundings, where memory effects become significant. Here, we extend the analysis of temporal correlations beyond the regression theorem, revealing what can be learned about open quantum systems when this hypothesis fails.
We introduce robust, operationally meaningful methods to explore how the breakdown of the regression hypothesis can uncover fundamental quantum features of non-Markovian open systems, including entanglement, coherence, and quantum memory, namely, the fundamental impossibility of simulating memory in non-Markov processes with classical feedback mechanisms. Finally, we demonstrate how these quantum features are linked to microscopic properties such as the bath spectral density and heat flow.
\end{abstract}
\maketitle

To learn about a quantum system, we must measure it and analyze the resulting data. The data acquisition frequently takes the form of a classical stochastic process, either as a continuous current or a discrete time series. Most information about a stochastic process is not found in its average value, but in the temporal fluctuations around it \cite{Gardiner1985}. For example, diffusion in complex systems is often characterized by two-time correlations through the Hurst exponent, which reveals anomalous diffusion, fractality, and complexity \cite{Mandelbrot1968}. In quantum mechanics, two-time correlations inform about Rabi oscillations \cite{Vijay2012}, relaxation rates, and equilibration in both open and isolated systems \cite{Alhambra2020,Dowling2023a,Dowling2023b}, and form the basis of Glauber’s theory of optical coherence \cite{Glauber1963}, among other phenomena. 

Computing two-time correlations in open quantum systems is challenging. In an optimistic scenario, one obtains a master equation governing the system's state evolution, from which expected values at individual points in time can be computed. The additional regression hypothesis (often called ``theorem'') posits that the evolution between consecutive times follows the same dynamical law as the system's evolution in the absence of measurement \cite{Lax68}, enabling the calculation of two- and multi-time correlations for a wide range of applications \cite{GardinerZoller2004,Landi2024}. However, central to the validity of the regression hypothesis is that the system couples weakly to its surroundings. More precisely, the hypothesis is formally equivalent to a strong notion of Markovianity or memorylessness \cite{Li2018}, which often fails to accurately reproduce multi-time correlations in non-Markovian processes; see Refs.~\cite{Ford1996,Alonso2005,Alonso2007,Guarnieri2014,VA17, Cosacchi2021,Kurt2021,Lonigro2022a,Lonigro2022b} for in-depth discussions.

Here we study temporal correlations from two-point measurement data on open quantum systems (see Fig.~\ref{fig:cover}), particularly when the regression hypothesis is violated. 
We categorize these correlations into two groups: those that can be represented as statistical mixtures of correlations satisfying the regression hypothesis, and those that cannot. 
We show that the violation of this extension of the regression hypothesis signals non-classicality by simultaneously revealing system-environment entanglement, quantum coherence and quantum memory effects. The latter refers to the impossibility of simulating the environment's memory using classical feedback mechanisms alone, a phenomenon that has been intensively studied in recent years \cite{GC21,MSXPMG21,NQAGMV21,Taranto2024,Roy2024,BBS24,Santos2024,Ohst2024,Goswami2024,Foligno2023,Carignano2024,Yu2025}.

In the framework of high-order quantum operations, quantum memory in non-Markov processes corresponds to an instance of separability problem with marginal constraints \cite{GC21,NQAGMV21}. Using this fact, here we propose a comprehensive method for witnessing and detecting non-classical temporal correlations through what we call ``entanglement retrievers''. They are defined operationally in terms of a channel discrimination task which involves memory-assisted quantum measurements (often called ``testers'') in order to ``recover'' or ``swap'' system-environment entanglement which is used as a resource. The corresponding non-classicality witnesses are then obtained from a hierarchy of semidefinite programs. 
A recent related approach employs a semidefinite program-based non-classicality criterion to construct similar witnesses \cite{Yu2025}.

\begin{figure}
    \centering
    \includegraphics[width=1\linewidth]{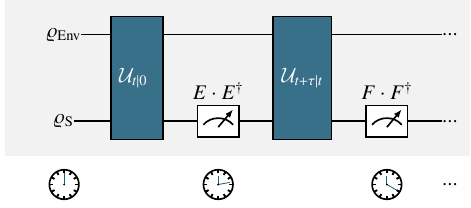}
    \caption{\textbf{Quantum two-point measurement process.} An open quantum system, initially prepared in a product state between the system and environment, is probed at successive points in time with arbitrary quantum instruments. We focus on the two-point measurement scenario, wherein the system is probed at times $t$ and $t+\tau$. In the considered scenario, all empirically accessible information can be expressed through correlations between operators $E$ and $F$; see Eq.~\eqref{eq:TPM}.}
    \label{fig:cover}
\end{figure}

Our quantum memory witnesses are, in a certain sense, complete and are not ``mere entanglement witnesses'' obtained through external approximations to the set of classical processes \cite{GC21}, nor do they involve calculation of complicated convex roofs and floors \cite{BBS24}. Furthermore, we show that these witnesses have an operational meaning in terms of quantum information backflow, taking into account the dynamic nature of the underlying processes---an aspect often overlooked in prior studies. Crucially, any witness that we derive for a given process tensor can be cast into an experimentally accessible two-time correlation functions between system operators in the Heisenberg picture, facilitating quantum memory characterization in practice, even when the process tensor is unknown. As an example, we construct a witness for the spin-boson model and demonstrate its capacity to detect quantumness of a memory. The witness admits a physical interpretation: classical temporal correlations impose a bound on the heat exchange rate between the system and the bath---a violation signals non-classicality.

\section{Two-point measurement correlations and quantum regression hypothesis}

Consider a quantum system $\rm S$ that interacts with uncontrolled degrees of freedom of an environment (Env). In the absence of interventions, the combined system $\rm S+Env$ is assumed to be isolated and undergoes unitary time evolution, described by operators $U(t|s)$ which result from time-ordered exponentiation of the total Hamiltonian from a time $s$ to a future time $t$. We denote by $\mathcal{U}_{t|s}[\cdot]=U(t|s)\cdot U(t|s)^\dagger$ the associated map that evolves density operators $\varrho_{\rm SEnv}(s)$ of the total system. The reduced system state at time $t$ will be $\varrho_{\rm S}(t)=\tr_{\rm Env} \hsp\mathcal{U}_{t|s}[\varrho_{\rm SEnv}(s)]$, where $\tr_{\rm Env}$ denotes partial trace over the environment. 

The scenario we envision involves two-point measurements (TPM), where the system $\rm S$ is locally probed at two successive points in time, $t$ and $t+\tau$, with \emph{arbitrary} quantum instruments (see Fig.~\ref{fig:cover}). TPM correlations can expressed as
\begin{equation}\label{eq:TPM}
\mathfrak{g}^{\textrm{\tiny$(2)$}}(E,F;t,\tau)=\langle E(t)^\dagger F(t+\tau)^\dagger F(t+\tau)E(t)\rangle,
\end{equation}
which equals the joint probability of observing the outcomes corresponding to the operators $E$ and $F$ at the times $t$ and $t+\tau$, respectively, with $E^\dagger E,F^\dagger F \leq \mathrm{id}$. Explicit time dependence indicates operators in the Heisenberg picture, and $\langle \cdot \rangle$ denotes expected value with respect to an initial product state for the total system, $\varrho_{\rm SEnv}(0) = \varrho_{\rm S} \otimes \varrho_{\rm Env}$. For example, if $E,F \propto b$, the Bosonic annihilation operator, then $\mathfrak{g}^{\textrm{\tiny$(2)$}}$ is proportional to the second-order coherence function \cite{GardinerZoller2004}.

Under the assumption of memoryless (or Markovian) evolution, the regression hypothesis offers a powerful simplification to calculate TPM correlations. Formally, a TPM process satisfies the regression hypothesis at times $t$ and $t+\tau$ if there exists a completely-positive (CP) and unital map $\hat{\Phi}_{t+\tau|t}$ [i.e., $\hat{\Phi}_{t+\tau|t}(\mathrm{id})=\mathrm{id}$, representing an \emph{arbitrary} channel in the Heisenberg picture] such that 
\begin{equation}\label{eq:QRT}
\mathfrak{g}^{\textrm{\tiny$(2)$}}(E,F;t,\tau)=\tr[E^\dagger\hat{\Phi}_{t+\tau|t}(F^\dagger F) \hsp E\hsp\varrho_{\rm S}(t)], \quad \forall \, E,F .
\end{equation}
In most cases where this hypothesis is invoked, the evolution of the system is described by a quantum dynamical semigroup, $\hat{\Phi}_{t+\tau|t}=\ee^{\tau\mathcal{L}^\dagger}$, with $\mathcal{L}$ the Lindblad generator \cite{Lindblad76,GKS76},
\begin{equation}\label{eq:Lindblad}
\mathcal{L} X=-\frac{\ii}{\hbar}[H_{\rm S},X]+\sum_{k}\gamma_k\frac{[J_kX,J_k^\dagger]+[J_k, XJ_k^\dagger]}{2},
\end{equation}
$J_k$ being the jump operator for the $k$th channel with corresponding decay rate $\gamma_k$, and $H_{\rm S}$ the (possibly Lamb-shifted) system Hamiltonian. 
More generally, the map $\hat{\Phi}_{t+\tau|t}$ may exhibit explicit time dependence, as in (non-)adiabatic Markovian master equations \cite{Albash2012,Yamaguchi2017,Dann2018,DiMeglio2024}.

The regression hypothesis requires a more stringent notion of memorylessness than those typically associated with quantum Markov processes in the literature, such as CP-divisilibity \cite{RHP14}. In fact, a finite-dimensional TPM process satisfies Eq.~\eqref{eq:QRT} for all $E$ and $F$ if and only if it is indistinguishable from (i.e., yields the same statistics as) a process where the system evolves through the application of \emph{independent} noisy quantum channels acting on the system \cite{Pollock2018a,Pollock2018b,Li2018,MM21}. 
In other words, on the timescale the system is interrogated, the environment’s effect is limited to adding noise to the system's evolution without retaining any memory.

A natural extension of the regression hypothesis consists of TPM correlations arising from statistical mixtures of elements that satisfy Eq.~\eqref{eq:QRT},
\begin{equation}\label{eq:GQRT}
\mathfrak{g}^{\textrm{\tiny$(2)$}}(E,F;t,\tau)=\int \mathfrak{g}^{\textrm{\tiny$(2)$}}_\lambda(E,F;t,\tau) \hsp \dd \omega_\lambda,
\end{equation}
where $\omega$ is a probability measure and $\mathfrak{g}^{\textrm{\tiny$(2)$}}_\lambda(E,F;t,\tau)$ satisfies Eq.~\eqref{eq:QRT} for all $\lambda$. As we will discuss below, this convex extension of the regression hypothesis encompasses open system dynamics with strong memory effects. Yet, it can be violated in some regimes, which, as we will argue, is an indication of non-classicality.

\section{Process matrices and testers}

The structure of TPM processes is most clearly represented with help of the state-channel duality provided by the Choi-Jamio{\l}kowki (CJ) isomorphism \cite{Choi75,J72} 
In this approach, we treat $\rm S$ as distinct systems with Hilbert spaces $\mathcal{H}_{\rm A}, \mathcal{H}_{\rm B}, \mathcal{H}_{\rm C}$ at different moments in time: we denote $\rm S$ as $\rm A$ at time $t$ before measuring, as $\rm B$ immediately after, and as $\rm C$ at time $t+\tau$ (see \hyperlink{Heisenberg_to_CJ}{Appendix A} for details). Eq.~\eqref{eq:TPM} becomes
\begin{equation}\label{eq:TPM2}
\mathfrak{g}^{\textrm{\tiny(2)}}(E,F;t,\tau)=\tr[(M_E\otimes F^\dagger F) W(t,\tau)].
\end{equation}
with $M_E$ the CJ representation
of the map $\mathcal{M}_{E}\cdot =E \cdot E^\dagger$ and $W(t,\tau)$ the so-called process matrix, 
\begin{equation}\label{eq:link_product}
W(t,\tau)^\mathrm{T}=\tr_{\rm Env}\{[\varrho_{\rm AEnv}(t)^{\TT_{\rm Env}}\otimes \mathrm{id}_{\rm BC}](\mathrm{id}_{\rm A}\otimes J_{t+\tau|t})\}.
\end{equation}
Here, $\mathrm{T}$ means the transposition and $\TT_{\rm Env}$ denotes a partial transposition over the environment, and $J_{t+\tau|t}$ is obtained from the CJ representation of the unitary map $\mathcal{U}_{t+\tau|t}$ by tracing out the environment at $t+\tau$.

The process matrix $W(t,\tau)$ is the analogue of the density operator for the TPM scenario, as it completely determines the measurement statistics and can be reconstructed via (process) tomography \cite{White2022}. 
Every TPM process matrix must be positive ($W\geqcurved 0$), normalized ($\tr W=\dim\mathcal{H}_{\rm B}$), and additionally satisfy $(\dim \mathcal{H}_{\rm B})\tr_{\rm C}\hsp W=\tr_{\rm BC}\hsp W\otimes \mathrm{id}_{\rm B}$ \cite{CDP08a,CDP08b,CDP09}. The first two conditions ensure valid probabilities while the third prohibits a retrocausal influence of an intervention at time $t+\tau$ on the statistics at time $t$.

Measurements can also be generalized to the level of TPM processes, giving rise to the notion of testers. These are protocols that probe the system at times $t$ and $t+\tau$, aided by an auxiliary memory system. Mathematically, a tester with outcomes in a finite alphabet $\mathcal{X}$ is defined as a family of operators $\{E_x\}_{x\in\mathcal{X}}$ such that $\tr(E_x W)$ gives the probability of obtaining outcome $x$ for a TPM process $W$. Positivity and normalization of probabilities demand that
the operators obey $E_x \geqcurved 0\hspace{0.15cm}\forall x\in\mathcal{X}$, $\sum_{x\in\mathcal{X}}E_x =(\dim \mathcal{H}_{\rm C})^{-1}\tr_{\rm C}\sum_{x\in\mathcal{X}}\hsp E_x\otimes \mathrm{id}$, and $(\dim \mathcal{H}_{\rm C})^{-1}\tr_{\rm BC}\hsp \sum_{x\in\mathcal{X}}E_x=\mathrm{id}$. 

\section{Results}

The primary advantage of formulating TPM correlations in the process matrix framework is that, in Eq.~\eqref{eq:TPM2}, the measurement operators are ``decoupled'' from the process, in contrast to Eq.~\eqref{eq:TPM}. Therefore, the breakdown of the decomposition in Eq.~\eqref{eq:GQRT} can be rephrased solely in terms of the process. This allows us to formulate our first main result (see \hyperlink{Proof_Obs1}{Appendix B} for the formal proof).

\begin{observation}\hypertarget{Observation1}{ }
TPM correlations satisfy the generalized quantum regression formula [Eq.~\eqref{eq:GQRT}] for all $E$ and $F$ if and only if the TPM process matrix can be written as
\begin{equation}\label{eq:CM}
W=\int \varrho_\lambda\otimes N_\lambda\hsp \dd \omega_\lambda,
\end{equation}
where $\omega$ is a probability measure, $\varrho_{\lambda}$ are density operators that average to (with respect to $\omega$) the actual reduced system state at $t$, and $N_{\lambda}$ are CJ representations of channels transforming system states from $t$ to $t+\tau$; that is $N_\lambda\geqcurved 0$ and $\tr_{\rm C}\hsp N_\lambda=\mathrm{id}_{\rm B}$ for all $\lambda$. Two sufficient conditions for a TPM process matrix to satisfy Eq.~\eqref{eq:CM} are: 
\begin{itemize}
    \item[(i)]\hypertarget{i}{ } the system-environment state $\varrho_{\rm SEnv}(t)$ is separable or
    \item[(ii)]\hypertarget{ii}{ } there exists an orthonormal basis $\{\ket{\mu}\}$ for the Hilbert space of the environment such that the maps
    \begin{equation}
    \tr_{\rm Env}\left[\mathcal{U}_{t+\tau|t}(\cdot \otimes \ket{\mu}\bra{\nu}) \right] \equiv 0 \quad \text{for} \quad \mu\neq\nu.
    \end{equation}
\end{itemize}
\end{observation}

So, the breakdown of \hyperlink{i}{(i)} \emph{and} \hyperlink{ii}{(ii)} is  necessary for a system-environment dynamics to violate the generalized quantum regression formula \eqref{eq:GQRT} in TPM processes:
the system must become entangled with the environment and coherences in \emph{all} bases of the environment must have an effect upon the system. 
Conversely, \hyperlink{ii}{(ii)} holds for strictly incoherent operations on the environment \cite{Yadin2016}, wherein the system-environment unitary can be written as 
\begin{equation}
U(t+\tau|t)=\sum_\mu V_\mu \otimes \ket{\mu}\bra{f(\mu)},
\end{equation}
Here, $V_\mu$ are arbitrary unitaries acting upon the system and $f$ is a bijection. Consequently, pure dephasing models are essentially classical \footnote{This is because any strictly incoherent operation, which includes pure dephasing, can be performed in terms of incoherent unitary operations (see Ref.~\cite{Yadin2016} for details).}, even those that exhibit strong memory effects such as the ones implemented in the single-photon experiments reported in Refs.~\cite{LLHLGLBP11,LLSLLGMP18}.

The decomposition~\eqref{eq:CM} of TPM process matrices was also obtained in Ref.~\cite{GC21} and referred to as processes with classical memory [and quantum memory (QM) otherwise]. Accordingly, we will write $W\in\mathbf{CM}$ for such processes. The characterization arises from the assumption that memory can be simulated through a \emph{classical} feedback mechanism, formally implemented by enacting measure-and-prepare (i.e., entanglement-breaking) channels upon the environment \footnote{It is therefore consistent with previous definitions of ``quantum memory'' in different contexts; cf. \cite{Nielsen2003,Chiribella2008,Lvovsky2009,Terhal2015,Rosset2018}.}. 

The convex-geometric structure of this set allows us to define operators which, akin to entanglement witnesses~\cite{Guehne2009}, provide sufficient criteria for QM. Formally, a QM witness $Z$ is any self-adjoint operator for which $W\in\mathbf{CM}$ implies $\tr(WZ)\geq 0$. Consequently, $\tr(ZW)<0$ indicates the presence of QM. However, while Eq.~\eqref{eq:CM} might seem equivalent to separability in the standard sense of a density matrix proportional to $W$ in the bipartition $\rm A:BC$, this is not the case \cite{NQAGMV21}. Each term $N_\lambda$ must not only be positive semidefinite but also satisfy the marginal constraint $\tr_{\rm C}\hsp N_\lambda=\mathrm{id}_{\rm B}$, owing to the \emph{dynamic} nature of TPM processes.

A central observation of this paper is that QM witnesses can be formulated in a complete and operationally relevant manner. To this end, we introduce ``entanglement retrievers" (ERs) as operators, $\Theta\in\mathbf{ER}$, that satisfy the semidefinite constraints (i) $\Theta\geqcurved 0$ and (ii) $ \eta \otimes \mathrm{id}_{\rm C}\geqcurved \tr_{\rm A}\hsp\Theta$ for some density operator $\eta$ on $\mathcal{H}_{\rm B}$. Physically, ERs act on the system input ($\rm A$) and output ($\rm C$) of the unitary $\mathcal{U}_{t+\tau|t}$ and lead to an output space isomorphic to $\rm A$; see Fig.~\ref{fig:ERs}.
First, we note that every ER gives rise to QM witnesses, 
see \hyperlink{Proof_Obs2}{Appendix C} for a formal proof.

\begin{observation}
If $W\in\mathbf{CM}$ then
\begin{equation}\label{eq:ER}
\mathcal{E}(W):=\max_{\Theta\in\mathbf{ER}} \tr(\Theta W)\leq 1. 
\end{equation}
Given $\Theta\in\mathbf{ER}$ and $Z_0$ positive on TPMs [i.e., $\tr(Z_0 W)\geq 0$ for all $W\in\mathbf{TPM}$], then
\begin{equation}\label{eq:QM-Wit}
Z(\Theta,Z_0)=Z_0+\kappa(\Theta,Z_0)\left(\Theta-\frac{\mathrm{id}}{\dim(\mathcal{H}_{\rm A}\otimes\mathcal{H}_{\rm B})}\right)
\end{equation}
is a QM witness, with
\begin{equation}
{\kappa(\Theta,Z_0)}=\min_{\Omega\in\mathbf{CM}}\frac{\tr(Z_0\Omega)}{(\dim\mathcal{H}_{\rm A})^{-1}-\tr(\Theta \Omega)}.
\end{equation}
Conversely, any QM witness can be brought into the form~\eqref{eq:QM-Wit}. 
\end{observation}

This result provides both a complete characterization of QM witnesses in terms of ERs and a simple method to check whether a given TPM process matrix $W$ exhibits QM: it does if $\mathcal{E}(W)>1$. Computing $\mathcal{E}(W)$ amounts to solving a semidefinite program (SDP), which can be done efficiently \cite{SC2023}. Furthermore, given an ER $\Theta_*$ obtained from the SDP of Eq.~\eqref{eq:ER} from some $W\in\mathbf{TPM}$, the classical memory threshold can be reduced from one to
\begin{equation}
\lambda_*(\Theta_*)=\max_{\Omega\in\mathbf{CM}}\tr(\Theta_* \Omega),
\end{equation}
which, in turn, is equivalent to computing $\kappa(Z_0,\Theta_*)$ when $Z_0$ satisfies $\tr(Z_0 W)=1$ for all $W\in \mathbf{TPM}$; see \hyperlink{Proof_Obs1}{Appendix B} for details. This maximization is no longer a SDP, but an instance of constrained bilinear optimization \cite{Berta2016}, which can be estimated to any required precision using the hierarchy of SDPs of Ref.~\cite{Berta2021}. Still, it can be expected to be of a similar complexity as the standard separability problem.

To see the role entanglement plays in defining ERs, we note that the channel-state duality allows us to write $\mathcal{E}(W)=(\dim\mathcal{H}_{\rm A})\hsp\bra{\Phi^+}\hsp\varrho_{\rm R}(W)\hsp\ket{\Phi^+}$, where $\ket{\Phi^+}\in\mathcal{H}_{\rm A}^{\otimes 2}$ is the normalized Bell state, $\varrho_{\rm R}(W) = \Theta_* \star W$
and $\star$ denotes the link product \footnote{The link product between two operators, say $X$ and $Y$ on Hilbert spaces $\mathcal{H}_{\rm A}\otimes \mathcal{H}_{\rm B}$ and $\mathcal{H}_{\rm B}\otimes \mathcal{H}_{\rm C}$ respectively, is defined as $$X\star Y=\tr_{\rm B}\big[\big(X^{\Gamma_{\rm B}}\otimes \mathrm{id}_{\rm C}\big)\big(\mathrm{id}_{\rm A}\otimes Y\big)\big].$$ It can be understood as a simple way to translate composition between channels in the CJ picture; see Ref.~\cite{CDP09} for details.}. The (non-normalized) bipartite state $\varrho_{\rm R}(W)$ can thus be seen as the most entangled state that one can create from the backflow of quantum information in a non-Markovian process with QM. In fact, notice that $\mathcal{E}(W)$ is upper-bounded by the singlet fraction of $\varrho_{\rm R}(W)$, a proper entanglement measure in bipartite systems \cite{Guehne2009}. It is upper-bounded by the state's Schmidt rank \cite{Terhal2000}, which, in turn, allows us to define a notion of QM \emph{dimension}: $\mathcal{E}(W)>d\geq 1$  means that the environment cannot be simulated by a $d$-dimensional quantum system, where classical memory corresponds to $d=1$. In addition, the connection to the singlet fraction
implies that $\varrho_{\rm R}(W)$ can only be detected by
an entanglement retriever, if it violates the criterion
of the positivity of the partial transpose.

Given a QM witness, it is natural to ask how to measure it by interrogating the system state at the two times $t$ and $t+\tau$. We provide two approaches; one based on a tester, which may require coupling to an auxiliary memory system, and one based on conventional TPM. 

The first of them gives the operational meaning to ERs, which also allows the certification of the presence of QM in a calibration-robust way \cite{Moroder2012}; 
that is, without requiring trust in the experimental apparatus for certification apart from knowing the system dimension. To do this, consider the following task: 
at time $t$, Alice uniformly samples a letter $x\in\mathcal{X}$ from an alphabet of size $|\mathcal{X}|=(\dim \mathcal{H}_{\rm A})^2$ and encodes it deterministically in the system $\rm S$ by a \emph{unital} CPTP map $\mathcal{E}_{x}$, whose output dimension does not exceed that of the input. Bob then tries to decode it using a tester, see Fig.~\ref{fig:ERs}. We denote $\mathcal{S}=(\{\mathcal{E}_x\}_{x\in\mathcal{X}},\{E_y\}_{y\in\mathcal{X}\cup \emptyset})$, where $E_x$ denotes the tester effect corresponding to input $\mathcal{E}_x$ and $E_{\emptyset}$ corresponds to a potentially inconclusive event. The set of all strategies is denoted by $\mathbf{Strat}$. Then, as shown in \hyperlink{AppendixD}{Appendix D}, one has:

\begin{observation}
Given a TPM process $W$, the optimal 
average probability of success is characterized by
\begin{equation}
\max_{\mathcal{S}\in\mathbf{Strat}}
\frac{1}{|\mathcal{X}|}\sum_{x\in\mathcal{X}}\mathrm{Pr}(y=x|x,\mathcal{S})
= \frac{\mathcal{E}(W)}{\dim \mathcal{H}_{\rm A}}
\end{equation}
where $\mathrm{Pr}(y=x|x,\mathcal{S})$ reads as the probability of correctly decoding input $x$ in a strategy $\mathcal{S}$.
\end{observation}

\begin{figure}
    \centering
    \includegraphics[width=1\linewidth]{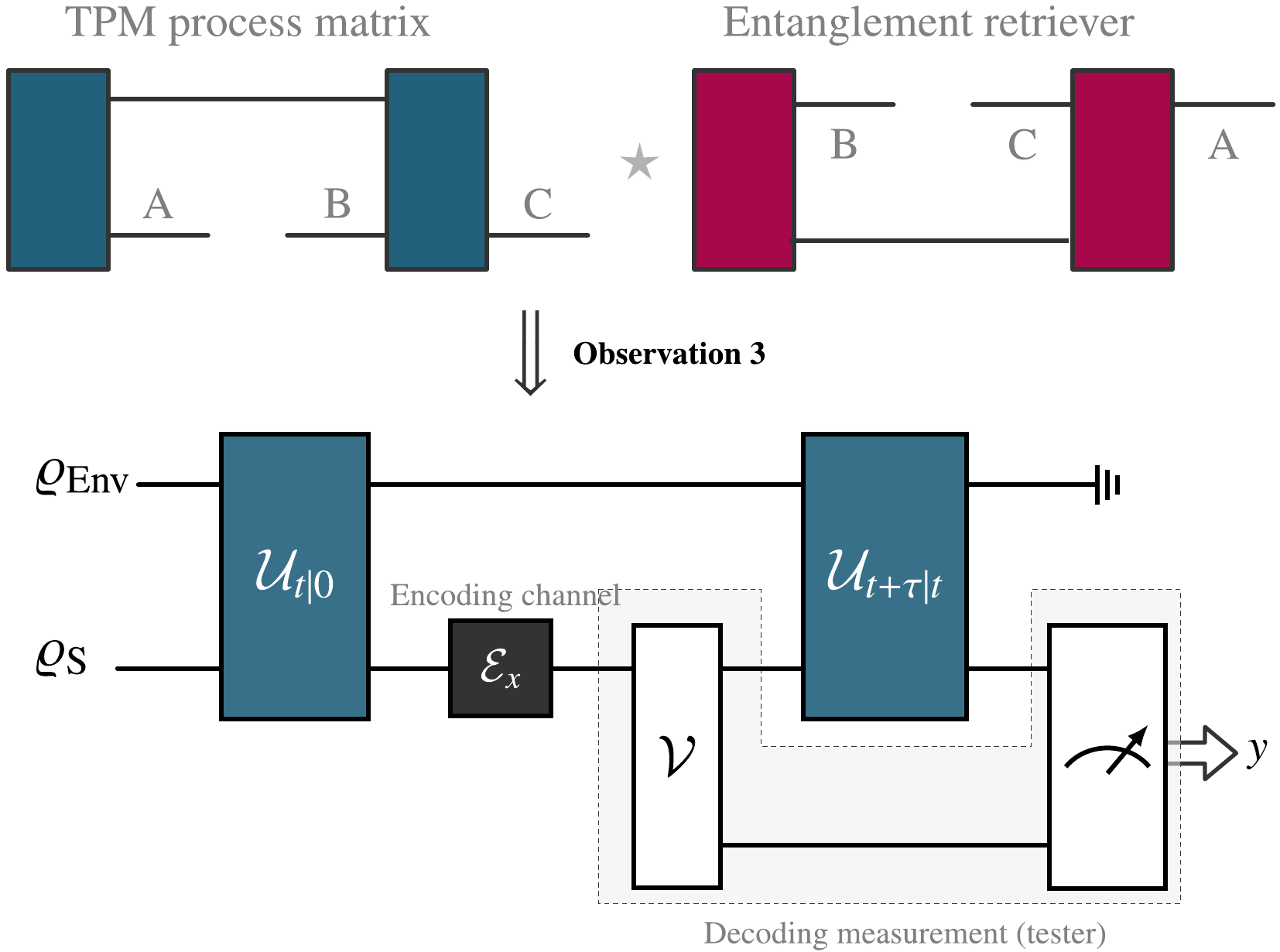}
    \caption{\textbf{Operational interpretation of ERs.} TPM process matrices and ERs correspond to high-order quantum operations, which can be represented as ``combs''. Their contraction ``$\star$'', $\Theta\star W=\tr(\Theta^\mathrm{T}\, W)$, is interpreted in terms of a superdense coding like protocol: an agent encodes $x\in\mathcal{X}$ with a unital channel $\mathcal{E}_x$, and another agent tries to decode it by measuring the system with a high-order measurement (tester), defined by a pre-processing operation $\mathcal{V}$, an auxiliary memory system, and a final joint measurement (POVM), which yields outcome $y$. }
    \label{fig:ERs}
\end{figure}

The second approach concerns a typical open system setting for which the process matrix $W$ is unknown. In this case, one may still be able to detect QM by measuring a suitable witness $Z$ by combining given TPM correlations. Any possible QM witness can 
be decomposed as a linear combination of TPM correlations, a fact which resembles the decomposition of entanglement witnesses into local measurements \cite{Guehne2009}.

\begin{observation}\hypertarget{Observation4}{ }
The action of a QM witness in the Heisenberg picture can always be written as
\begin{equation}\label{eq:memory_Heisenberg}
\tr[Z W(t,\tau)]=\sum_{ij}d_{ij}(Z)\hsp\mathfrak{g}^{\textrm{\tiny$(2)$}}(E_i,F_j;t,\tau),
\end{equation}
where $d_{ij}(Z)$ are real coefficients, the operators $E_i$ and $F_j$ are fixed, i.e., at the initial time they do not depend on the process nor the witness.
\end{observation}

The proof is given in \hyperlink{AppendixE}{Appendix E}.

\section{Quantum memory detection in cavity QED}

Here, we illustrate our QM detection scheme with a physical example. It consists of a two-level system, which is linearly coupled to a bath according to the Hamiltonian
\begin{equation}
H=\frac{\hbar  \omega_0}{2}\hsp\sigma_z+H_{\rm Env}+\sigma\otimes B^\dagger+\sigma^\dagger \otimes B.
\end{equation}
The system $\rm S$ is characterized by its ground and excited states, $\ket{g}$ and $\ket{e}$, which are separated by an energy-gap of $\hbar \omega_0$ and $\sigma:=\ket{e}\bra{g}$. The environment is assumed to be made of non-interacting Bosonic modes, 
\begin{equation}
H_{\rm Env}=\sum_k \hbar \omega_k b_k^\dagger b_k^{ }\hsp\textrm{ and }\hsp B=\sum_k\hbar g_k b_k^{ },
\end{equation}
where $b_k$ is the lowering operator for the $k$th mode with frequency $\omega_k$ and coupling constant $g_k$.

In this setting, the process matrix cannot be obtained exactly for a generic initial state of the bath, except for the case of a single bath mode.
In summary, we proceed as follows. First, we obtain the process matrix in the 
single-mode case.
Using our \hyperlink{Observation2}{Observation 2}, we derive an ER, which, fortunately, is simultaneously optimal [in the sense that there is one $W\in\mathbf{CM}$ with $\tr(W \Theta_*)=1$] and time-independent,
\begin{equation}
\Theta_*=2\hsp\ket{g}\bra{g}\otimes \ket{\Psi^-}\bra{\Psi^-},
\end{equation}
where $\ket{g}\in\mathcal{H}_{\rm B}$ and 
$\ket{\Psi^-}\propto\ket{g}\otimes \ket{e}-\ket{e}\otimes \ket{g} \in\mathcal{H}_{\rm A}\otimes \mathcal{H}_{\rm C}$ is the singlet state. Next, we employ our \hyperlink{Observation4}{Observation 4}, using the Heisenberg equations of motion as a power series on $\tau$. At the end, we arrive at the condition for non-quantumness
\begin{align}\label{eq:inequality}
2\hbar \omega_0\mathrm{Re}\Big\langle&{N}(t)^{-\frac{1}{2}}\sin\big({N}(t)^\frac{1}{2}\tau\big){B}(t)^\dagger {\sigma}(t)\cos\big({N}(t)^\frac{1}{2}\tau\big)\Big\rangle\nonumber \\ &+\Big\langle\cos\Big(2{N}(t)^{\frac{1}{2}}\tau\Big){H}_{\rm S}(t)\Big\rangle\leq \frac{\hbar \omega_0}{2},
\end{align}
where $N=B^\dagger B$ and $2H_{\rm S}=\hbar  \omega_0\sigma_z$. This inequality is general, and does not depend on the initial bath state nor the spectral density. Among other things, this inequality places a limit on the system-bath energy flow, that is, heat, in the absence of quantum memory. This can be observed by Taylor expanding the left-hand side in $\tau$; e.g., to first order one has $E(t)+\tau \dot{E}(t)+\mathcal{O}(\tau^2)$, $E(t)=\langle H_{\rm S}(t)\rangle$.

\begin{figure}
    \centering
    \includegraphics[width=1\linewidth]{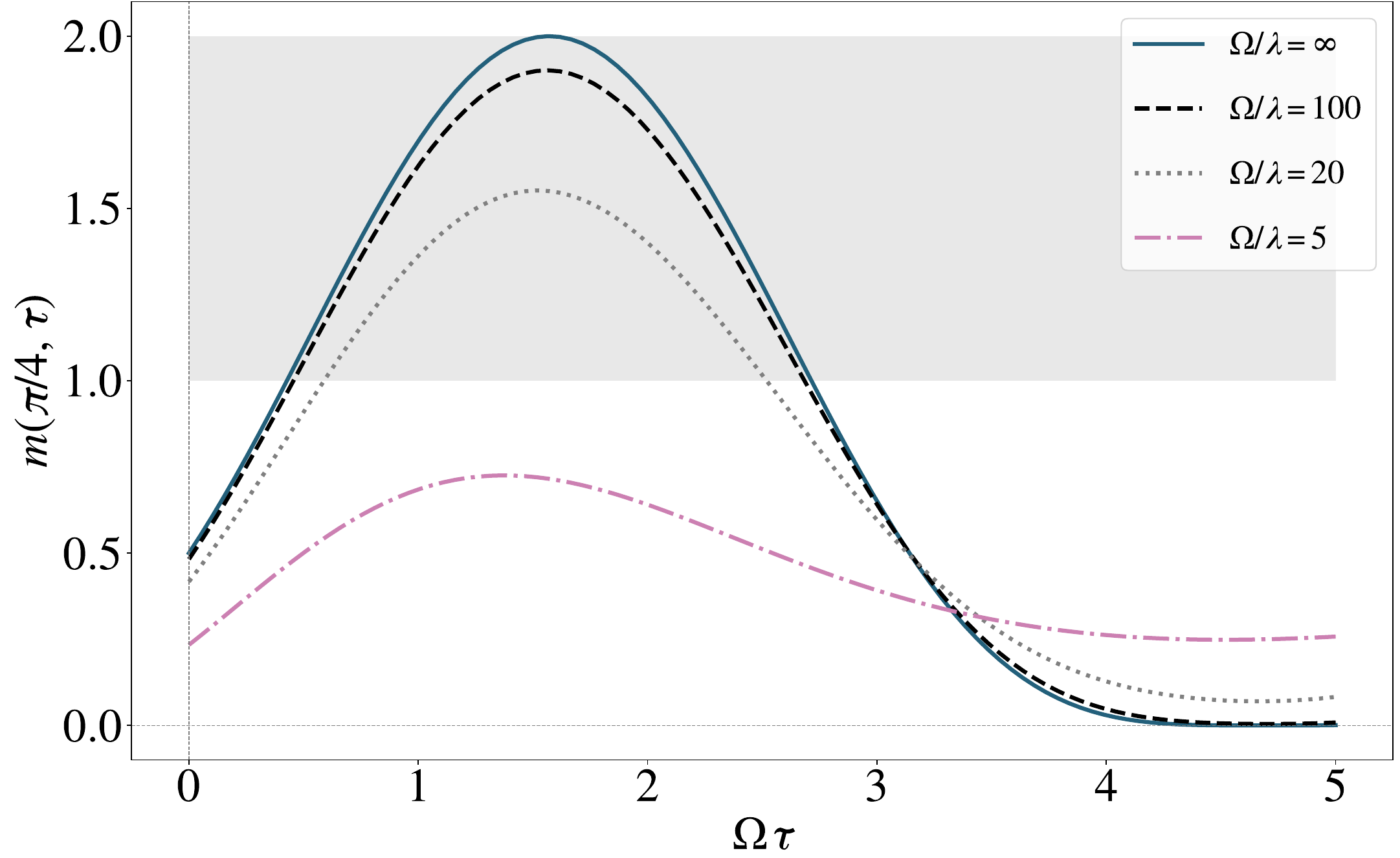}
    \caption{\textbf{Quantum memory in a damped resonant cavity.} Breakdown of the generalized regression formula (gray shade) for a two-level atom within a resonant optical cavity (Lorentzian spectral density) with frequency $\omega_0$, damping rate (inverse quality factor) $\lambda$ and Rabi frequency $\Omega$. The strong (non-Markovian) coupling regime occurs whenever $\Omega/\lambda<1$. Yet, this does not necessarily imply that the presence of 
    a quantum memory is detected (magenta curve) with the given witness. 
    }
    \label{fig:plot}
\end{figure}

The effect of the bath spectral density shows up explicitly if
we consider the system initially in the excited state and the modes in the vacuum state. The system-bath state remains pure at all times and \eqref{eq:inequality} simplifies to $m(t,\tau)\leq 1$, with
\begin{align}\label{eq:memory_SD}
m(t,\tau)&=\left|q(t)+\int_0^t\int_0^\tau f(t+\tau-r-s)q(r)q(s)\dd r \dd s\right|^2,
\end{align}
and $q(t)$ the probability amplitude of the initial state at time $t\geq 0$.
The second term in \eqref{eq:memory_SD} contains 
the two-time correlation function, $f(u)=\tr[B(u)B\varrho_{\rm Env}]\ee^{\ii  \omega_0 u}$, which is the Fourier transform of the bath spectral density $J(\omega)$, 
\begin{align}
f(u)=\int J(\omega)\ee^{\ii ( \omega_0-\omega)u}\hsp\dd\omega.
\end{align}
Fig.~\ref{fig:plot} shows exemplary results for exponentially decaying correlations, $2f(u)=\gamma_0 \lambda \ee^{-\lambda |u|}$, which correspond to a Lorentzian spectral density as in the case of a two-level atom within a resonant optical cavity \cite{BP07}. 
We observe that non-Markovianity in the standard sense (i.e., the breakdown of CP-divisibility) does not necessarily imply quantumness. Vacuum Rabi oscillations at the frequency $\Omega = \sqrt{2\gamma_0\lambda -\lambda^2}/2$ are observable in the regime of strong coupling and high cavity finesse, $\gamma_0 \gg \lambda$. The threshold for quantum memory detection lies at $\Omega \gtrsim 7\lambda $. The seminal experiment~\cite{Brune1996}, achieved $\Omega \approx 10^2 \lambda$, safely above that threshold. 

Finally, we remark that Eq.~\eqref{eq:memory_SD} holds true for \emph{all} spectral densities. In \hyperlink{appendixF}{Appendix F} we analyse in some details the Ohmic case, 

\section{Discussion}

Temporal correlations are essential in quantum physics, influencing both foundational questions and a diverse array of applications \cite{Budroni2013,Budroni2014,Budroni2016,Hoffmann2018,Costa2018,Spee2020,Budroni2021,Budroni2021,Vieira2024}. In this context, we enhance the analysis of these correlations in open systems, illustrating how they can indicate non-classicality through manifestations such as entanglement, coherence, and quantum memory within the framework of two-point measurements.

Modeling quantum processes with high-order maps allows us to access information that cannot be obtained from the reduced state of the open system, in particular analyze quantum temporal correlations. However, obtaining process tensors and extracting information from temporal correlations is highly challenging, with a variety of approaches proposed in recent years to address this problem; see, e.g., Refs.~\cite{Strathearn2018,Binder2018,J2019,Cygorek2022,White2023,Dowling2024,Cygorek2024a,Cygorek2024b,Li2025} and references therein. Here, we present an approach that can circumvent the necessity of knowing the process tensor exactly to detect non-classicality and applied it to a physically relevant model.

A possible future step would be to extend the analysis to multi-time correlations (i.e., involving more than two time slots). Although we believe that most of the results established here can be straightforwardly generalized, a careful analysis could unveil new physical insights that are not apparent in two-point correlations. 
Indeed, the authors of Ref.~\cite{Taranto2024} demonstrate that generalizing the concept of quantum memory to multi-time scenarios is considerably more complex; however, whether this reveals new physical insights (and if so, which ones) remains an intriguing open question.

\section{Acknowledgements}

We thank Luis A. Correa, Edward Gandar, Satoya Imai, Ties-A Ohst, Jyrki Piilo, Marco Tulio Quintino, Zhen-Peng Xu and Benjamin Yadin for the discussions and key remarks. This work was supported by the Deutsche Forschungsgemeinschaft (DFG, German Research Foundation, project numbers 447948357 and 440958198), the Sino-German Center for Research Promotion (Project M-0294), and the German Ministry of Education and Research (Project QuKuK, BMBF Grant No. 16KIS1618K). LSVS acknowledges support from the House of Young Talents of the University of Siegen.

%\bibliography{ref}


\setcounter{equation}{0}
\setcounter{thm}{0}
\renewcommand{\theequation}{A\arabic{equation}}
\renewcommand{\thefigure}{A\arabic{figure}}
\setcounter{page}{1}
\makeatletter

\setcounter{secnumdepth}{2}

\twocolumngrid

\section*{Appendix A. From Heisenberg to Choi-Jamio{\l}kowski and back}\hypertarget{Heisenberg_to_CJ}{ }

Herein, we aim to succinctly explain how to transition between the Heisenberg and Choi-Jamio{\l}kowski (CJ) pictures for two-point measurement (TPM) correlations. This connection can be more easily understood by considering TPM correlations in the Schr\"{o}dinger picture first, wherein
\begin{equation}\label{eq:Born_Schrodinger_SM}
\mathfrak{g}^{\textrm{\tiny(2)}}(E,F;t,\tau)=\tr[F^\dagger F\hsp\mathcal{U}_{t+\tau|t}\circ\mathcal{M}_{E}\circ \mathcal{U}_{t|0}\hsp \varrho_{\rm S}\otimes \varrho_{\rm Env}].
\end{equation}
Here, $\mathcal{M}_{E}[\cdot] =E \cdot E^\dagger$ and $\varrho_{\rm S}\otimes \varrho_{\rm Env}$ is the initial system-environment state. Although this is perhaps the most intuitive way to understand, it is also the most inconvenient to perform any kind of calculations, whether analytical or numerical. Because of this, we move on to more convenient pictures, namely to the Heisenberg and CJ pictures.

To write Eq.~\eqref{eq:Born_Schrodinger_SM} in the Heisenberg picture we have to take the Hilbert-Schmidt dual of the maps. For this, it is convenient to write $\mathcal{U}_{t|s}$ as $\mathcal{U}_{t|s}\cdot =U(t|s)\cdot U(t|s)^\dagger$ and notice that $U(t+\tau|t)=U(t+\tau|0)U(t|0)^\dagger$ whence
\begin{align}
\mathfrak{g}^{\textrm{\tiny(2)}}(E,F;t,\tau)&=\tr[F^\dagger F\hsp\mathcal{U}_{t+\tau|t}\circ\mathcal{M}_{E}\circ \mathcal{U}_{t|0}\hsp \varrho_{\rm S}\otimes \varrho_{\rm Env}]\nonumber \\ &=\tr\big[\mathcal{U}_{t|0}^\dagger\circ\mathcal{M}_{E}^\dagger \circ\mathcal{U}_{t+\tau|t}^\dagger(F^\dagger F)\hsp\varrho_{\rm S}\otimes \varrho_{\rm Env}\big] \nonumber \\ &=\tr\big[\mathcal{U}_{t|0}^\dagger(E^\dagger)\mathcal{U}_{t+\tau|t}^\dagger(F^\dagger F)\mathcal{U}_{t|0}^\dagger(E)\hsp\varrho_{\rm S}\otimes \varrho_{\rm Env}\big]\nonumber \\ &=\langle E(t)^\dagger F(t+\tau)^\dagger F(t+\tau)E(t)\rangle.
\end{align}
The Heisenberg picture version of an operator $X$ (in 
Schr\"{o}dinger picture) is defined as $X(t)=\mathcal{U}_{t|0}^\dagger \hsp (X)$.

To move to the CJ picture, it is convenient to represent the open system $\rm S$ distinctly at distinct instants of time. We denote $\rm S$ as $\rm A$ at time $t$, before it is measured; after it is measured, the output system will be denoted as $\rm B$. The measurement action associated to $E$ is described by the CP linear map $\mathcal{M}_{E}:\mathrm{L}(\mathcal{H}_{\rm A})\to\mathrm{L}(\mathcal{H}_{\rm B})$, $\mathcal{M}_{E}[\cdot] = E\cdot E^\dagger$. The channel-state duality provided by the CJ isomorphism \cite{Choi75,J72} allows us to represent it as an operator $M_{E}\in\mathrm{L}(\mathcal{H}_{\rm A}\otimes \mathcal{H}_{\rm B})$ which satisfies $M_{E}\geqcurved 0$ and $\mathrm{id}_{\rm A}\geqcurved \tr_{\rm B}\hsp M_E$, obtained from the map $\mathcal{M}_{E}$ via
\begin{equation}\label{eq:CJI}
M_{E}=\sum_{ij}\ket{i}\bra{j}_{\rm A}\otimes \mathcal{M}_{E}(\ket{i}\bra{j}_{\rm A}), \end{equation} 
where $\{\ket{i}_{\rm A}\}\subset \mathcal{H}_{\rm A}$ is a \emph{fixed} orthonormal basis. The inverse mapping reads
\begin{equation}
\mathcal{M}_{E}\cdot =\tr_{\rm A}[(\hsp\cdot^\mathrm{T}\otimes \mathrm{id}_{\rm A})M_{E}],
\end{equation}
the transposition $\mathrm{T}$ being with respect to the bases we have fixed. Finally, the system $\rm S$ at the time $t+\tau$ is denoted as $\rm C$, thus $F\in\mathrm{L}(\mathcal{H}_{\rm C})$, $F^\dagger F\leq \mathrm{id}_{\rm C}$. 

The channel-state duality provided by CJ isomorphism together with Eq.~\eqref{eq:Born_Schrodinger_SM} imply that TPM probabilities are bilinear functions of the measurement operators $M_{E}$ and $F^\dagger F$,
\begin{equation}
\mathfrak{g}^{\textrm{\tiny(2)}}(E,F;t,\tau)=\mathcal{T}(M_{E},F^\dagger F).
\end{equation}
The map $\mathcal{T}$ is so-called process tensor, and encodes the statistically relevant aspects of system-environment interactions. Invoking Riesz representation lemma we can write
\begin{equation}
\mathfrak{g}^{\textrm{\tiny(2)}}(E,F;t,\tau)=\tr[(M_{E}\otimes F^\dagger F)W].
\end{equation}
Next, we impose $W$ to be positive semidefinite in order to ensure positivity of probabilities \footnote{To be fair, at this point one would only need to impose that $W$ be ``positive over tensor products'' to guarantee positivity of the probabilities (see Refs.~\cite{OCB12,ABCFGB15} for in-depth discussions). The need for $W$ to be \emph{de facto} positive semidefinite comes from the possibility of probing the system $\rm S$ with higher-order operations involving memory and entanglement.}. The causal order of the process, or, more precisely, the assumption of non-existence of retrocausal influences, together with normalization of probabilities imply $\tr_{\rm C}\hsp W=(\dim\mathcal{H}_{\rm B})^{-1}\tr_{\rm BC}\hsp W\otimes \mathrm{id}_{\rm B}$ and $\tr W=\dim\mathcal{H}_{\rm B}$ respectively \cite{CDP09}. 

In the context of OQSs, the TPM process matrix can be written in terms of the unitary maps and system-environment state as
\begin{equation}\label{eq:link_productApp}
W(t,\tau)^\mathrm{T}=\tr_{\rm Env}\{[\varrho_{\rm AEnv}(t)^{\TT_{\rm Env}}\otimes \mathrm{id}_{\rm BC}](\mathrm{id}_{\rm A}\otimes J_{t+\tau|t})\}.
\end{equation}
Here, $J_{t+\tau|t}$ is obtained from the unitary map $\mathcal{U}_{t+\tau|t}$ via CJ isomorphism [Eq.~\eqref{eq:CJI}] when the output space corresponding to the environment is traced out. Specifically, define $\mathcal{U}_{t+\tau|t}:\mathrm{L}(\mathcal{H}_{\rm B}\otimes \mathcal{H}_{\rm Env})\to \mathrm{L}(\mathcal{H}_{\rm B}\otimes \mathcal{H}_{\rm Env^\prime})$, where $\mathcal{H}_{\rm Env^\prime}\simeq \mathcal{H}_{\rm Env}$, and $\mathcal{J}(U_{t+\tau|t})\in\mathrm{L}(\mathcal{H}_{\rm A}\otimes \mathcal{H}_{\rm B}\otimes \mathcal{H}_{\rm Env}\otimes \mathcal{H}_{\rm Env^\prime})$, the latter obtained from $\mathcal{U}_{t+\tau|t}$ via Eq.~\eqref{eq:CJI}, wherein $J_{t+\tau|t}\equiv\tr_{\rm Env^\prime}\hsp \mathcal{J}(U_{t+\tau|t})$.

\section*{Appendix B. Proof of Observation 1}\hypertarget{Proof_Obs1}{ }

\begin{obsapp}\hypertarget{Observation1}{ }
TPM correlations satisfy the generalized quantum regression formula for all $E$ and $F$ if and only if the TPM process matrix can be written as
\begin{equation}\label{eq:CM_SM}
W=\int \varrho_\lambda\otimes N_\lambda\hsp \dd \omega_\lambda,
\end{equation}
where $\omega$ is a probability measure, $\varrho_{\lambda}$ are density operators that average to (with respect to $\omega$) the actual reduced system state at $t$, and $N_{\lambda}$ are CJ representations of channels transforming system states from $t$ to $t+\tau$; that is $N_\lambda\geqcurved 0$ and $\tr_{\rm C}\hsp N_\lambda=\mathrm{id}_{\rm B}$ for all $\lambda$. Two sufficient conditions for a TPM process matrix to satisfy Eq.~\eqref{eq:CM_SM} are: 
\begin{itemize}
    \item[(i)]\hypertarget{i}{ } the system-environment state $\varrho_{\rm SEnv}(t)$ is separable or
    \item[(ii)]\hypertarget{ii}{ } there exists an orthonormal basis $\{\ket{\mu}\}$ for the Hilbert space of the environment such that the maps
    \begin{equation}
    \tr_{\rm Env}\left[\mathcal{U}_{t+\tau|t}(\cdot \otimes \ket{\mu}\bra{\nu}) \right] \equiv 0 \quad \text{for} \quad \mu\neq\nu.
    \end{equation}
\end{itemize}
\end{obsapp}
\begin{proof}
To prove the equivalence between \eqref{eq:CM_SM} and the generalized quantum regression formula, it is sufficient to prove that Markov processes matrices, i.e., those satisfying $W_{\rm Markov}=\varrho\otimes N$, are in the same foot as the quantum regression theorem. Although the equivalence between such kind of quantum Markov processes and the quantum regression theorem has been considered before 
(see, e.g., Refs.~\cite{MM21} and \cite{Li2018}), for completeness, we will revisit such a result in our own way, from which the generalization [Eq.~\eqref{eq:CM_SM}] generalization will follow trivially from convex linearity.

To start off, note that such quantum Markov processes satisfy the standard quantum regression formula:
\begin{align}
\mathfrak{g}^{\textrm{\tiny(2)}}(E,F;t,\tau)&=\tr[(M_{E}\otimes F^\dagger F)(\varrho\otimes N)] \nonumber \\ &=\tr[F^\dagger F \mathcal{N}\circ \mathcal{M}_E(\varrho)] \nonumber \\ &=\tr[E^\dagger \mathcal{N}^\dagger(F^\dagger F)E\hsp \varrho],
\end{align}
where $N$ and $\mathcal{N}$ are CJ-dual of each other. Conversely, consider
\begin{equation}
\mathfrak{g}^{\textrm{\tiny(2)}}(E,F;t,\tau)=\tr[E^\dagger \mathcal{N}^\dagger(F^\dagger F)E\hsp \varrho],
\end{equation}
holding for all $E$ and $F$. Moving to the CJ picture we obtain
\begin{equation}
\tr[(X\otimes Y)(\varrho\otimes N)]=\tr[(X\otimes Y)W],
\end{equation}
for all $X\in\mathrm{L}(\mathcal{H}_{\rm A}\otimes \mathcal{H}_{\rm B})$, $Y\in\mathrm{L}(\mathcal{H}_{\rm C})$, positive semidefinite, where $X$ and $Y$ are obtained from $E$ and $F$ via CJ isomorphism [Eq.~\eqref{eq:CJI}]. The equality $W=\varrho\otimes N$ follows from the possibility of defining operator basis with positive semidefinite operators, like informationally complete POVMs.
From convex linearity, Eq.~\eqref{eq:CM_SM} holds true if and only if the TPM correlations are convex mixtures of quantum regression-like correlations.

To prove that conditions both \hyperlink{i_sm}{(i)} and \hyperlink{ii_sm}{(ii)} imply Eq.~\eqref{eq:CM_SM}, the following Lemmas will be useful.

\begin{lem}\hypertarget{Lemma1}{ }
Let $U:\mathcal{H}_1\otimes \mathcal{H}_2\to \mathcal{H}_1\otimes \mathcal{H}_2$ be a unitary operator. Then, for any orthonormal basis $\mathcal{B}_2=\{\ket{\mu}:\ket{\mu}\in\mathcal{H}_2\}$, one has
\begin{equation}
U(\ket{\psi}\otimes \ket{\phi})=\sum_\mu L_{\mu|\phi}\ket{\psi}\otimes \ket{\mu},
\end{equation}
where, for all fixed norm-1 vectors $\ket{\phi}\in\mathcal{H}_2$, $L_{\mu|\phi}$ are Kraus-like operators, i.e., $\sum_\mu L_{\mu|\phi}^\dagger L_{\mu|\phi}^{ }=\mathrm{id}$, that define a quantum instrument with elements
\begin{equation}\label{eq:instrumentL}
\mathcal{L}_{\mu|\phi}[\cdot] =\tr_{2}[U(\cdot \otimes \ket{\phi}\bra{\phi})U^\dagger(\mathrm{id} \otimes \ket{\mu}\bra{\mu})].
\end{equation}
\end{lem}
\begin{proof}
Let $\mathcal{B}_1=\{\ket{j}:\ket{j}\in\mathcal{H}_1\}$ be an orthonormal basis and write
\begin{align}\label{eq:Kraus}
U(\ket{i}\otimes \ket{\hat{\mu}})&=\sum_{j\mu} \mathfrak{d}_{j\mu}^{i\hat{\mu}}\ket{j}\otimes \ket{\mu}\nonumber \\ &=\sum_{\mu}\left(\sum_j \mathfrak{d}_{j\mu}^{i\hat{\mu}}\ket{j}\right)\otimes \ket{\mu}\nonumber \\ &=\sum_\mu L_{\mu|\hat{\mu}}\ket{i}\otimes \ket{\mu}
\end{align}
where $\mathfrak{d}_{j\mu}^{i\hat{\mu}}$ are the coordinates of the transformed vector on the basis $\mathcal{B}_1\otimes \mathcal{B}_2$. The operators $L_{\mu|\hat{\mu}}$ are defined via
\begin{equation}
L_{\mu|\hat{\mu}}\ket{i}:=\sum_j \mathfrak{d}_{j\mu}^{i\hat{\mu}}\ket{j}.
\end{equation}
This operator is independent on the basis $\mathcal{B}_1$. Similarly, one can define $L_{\mu|\phi}:=\sum_{\hat{\mu}} \langle \hat{\mu}|\phi\rangle L_{\mu|\hat{\mu}}$, for any $\ket{\phi}\in\mathcal{H}_2$.

The fact that $L_{\mu|\phi}$ are Kraus-like operators follow from Eq.~\eqref{eq:Kraus}, observing that $U$ is unitary,
\begin{equation}
\delta_{ij}=\sum_\mu \bra{i}L_{\mu|\phi}^\dagger L_{\mu|\phi}^{ }\ket{j}.
\end{equation}
Finally, denote by $\mathcal{L}_{\mu|\phi}[\cdot] =L_{\mu|\phi}^{ }\cdot L_{\mu|\phi}^\dagger$ the CP map induced by $L_{\mu|\phi}^{ }$, which acts as
\begin{align}
\mathcal{L}_{\mu|\phi}(\ket{i}\bra{j})&:=L_{\mu|\phi}^{ } \ket{i}\bra{j} L_{\mu|\phi}^\dagger \nonumber \\ &=(\mathrm{id}\otimes \bra{\mu})U(\ket{i}\bra{j}\otimes \ket{\phi}\bra{\phi})U^\dagger(\mathrm{id}\otimes\ket{\mu})\nonumber \\ &=\tr_{2}[U(\ket{i}\bra{j}\otimes \ket{\phi}\bra{\phi})U^\dagger(\ket{\mu}\bra{\mu}\otimes \mathrm{id})].
\end{align}
Then, Eq.~\eqref{eq:instrumentL} follows from linearity.
\end{proof}

\begin{lem}\hypertarget{Lemma2}{}
A TPM process with initial product state $\varrho_{\rm SEnv}(0)=\varrho_{\rm S}\otimes \varrho_{\rm Env}$ and evolution governed by unitary maps $\mathcal{U}_{t+\tau|t}$ yields the following process matrix
\begin{equation}\label{eq:decompositionW}
W(t,\tau)=\sum_{\mu\nu}\tr_{\rm Env}[\varrho_{\rm SEnv}(t)(\mathrm{id}\otimes P_{\mu\nu})]\otimes N_{\mu\nu}(t,\tau),
\end{equation}
where the vectors $\ket{\mu}\in \mathcal{H}_{\rm Env}$ define an orthonormal basis, $P_{\mu\nu}\equiv\ket{\mu}\bra{\nu}$, and $N_{\mu\nu}(t,\tau)$ is the CJ-dual of the map
\begin{equation}
\mathcal{N}_{\mu\nu}(t,\tau)[\cdot]=\tr_{\rm Env}\circ\mathcal{U}_{t+\tau|t}(\cdot \otimes P_{\mu\nu}).
\end{equation}
\end{lem}
\begin{proof}
It suffices to consider an initial pure state of the form $\ket{\psi}\otimes \ket{\phi}$; the extension for a generic product density operator will follows from convex linearity. Under this assumption, the TPM process matrix is obtained by tracing-out the following four-partite (non-normalized) state,
\begin{equation}
\|\kappa(t,\tau)\rangle=\mathcal{U}_{t+\tau|t}\big(\mathbb{S}_{\rm AC}\hsp \mathcal{U}_t(\ket{\psi}_{\rm A}\otimes \ket{\phi})\otimes \|\Phi^+\rangle_{\rm BC}\hsp \mathbb{S}_{\rm AC}\big)
\end{equation}
where $\mathbb{S}_{\rm AC}$ is the unitary \textsc{swap} on $\mathcal{H}_{\rm A}\otimes\mathcal{H}_{\rm C}$ and $\|\Phi^+\rangle_{\rm BC}$ the non-normalized Bell state, 
\begin{equation}
\|\Phi^+\rangle=\sum_i \ket{i}\otimes \ket{i}.
\end{equation}
Next, apply \hyperlink{Lemma1}{Lemma 1} to it twice to arrive at
\begin{equation}
\|\kappa(t,\tau)\rangle=\sum_{i\mu \mu^\prime} (L_{\mu|\phi}\ket{\psi})\otimes \ket{i}\otimes (K_{\mu^\prime|\mu}\ket{i})\otimes \ket{\mu^\prime},
\end{equation}
where $L_{\mu|\phi}$ and $K_{\mu^\prime|\mu}$ are Kraus-like operators. Tracing out the environment, one gets
\begin{align*}
W(t,\tau)&=\tr_{\rm Env}\hsp \|\kappa(t,\tau)\rangle\langle\kappa(t,\tau)\|\nonumber \\ &=\sum_{\ij\mu\nu}\tr_{\rm Env}[\mathcal{U}_{t|0}(\ket{\psi}\bra{\psi}\otimes \ket{\phi}\bra{\phi})(\mathrm{id}\otimes P_{\mu\nu})]\otimes \ket{i}\bra{j}\nonumber \\ &\hspace{3.25cm}\otimes \tr_{\rm Env}[\mathcal{U}_{t+\tau|t}(\ket{i}\bra{j}\otimes P_{\mu\nu})],
\end{align*}
where Eq.~\eqref{eq:instrumentL} of \hyperlink{Lemma1}{Lemma 1} was used. Eq.~\eqref{eq:decompositionW} follows from convex-linearity together with the CJ isomorphism [Eq.~\eqref{eq:CJI}].
\end{proof}

The decomposition of Eq.~\eqref{eq:decompositionW} in \hyperlink{Lemma2}{Lemma 2} allows us to state the connection between entanglement, coherence and quantum memory: \\

\noindent  (i) Suppose that at the time $t$ the system-environment state, $\varrho_{\rm Env}(t)$, is separable,
\begin{equation}
\varrho_{\rm SEnv}(t)=\int \varrho_\lambda \otimes \sigma_\lambda \hsp \dd\omega_\lambda,
\end{equation}
where $\varrho_\lambda$ and $\sigma_\lambda$ are density operators and $\omega$ is a probability measure. Plugging this into Eq.~\eqref{eq:decompositionW} we obtain
\begin{align}
W(t,\tau)&=\int \varrho_\lambda\otimes \sum_{\mu\nu}\tr(\sigma_\lambda P_{\mu\nu}) N_{\mu\nu}(t,\tau)\hsp \dd \omega_\lambda\nonumber \\ &=\int \varrho_\lambda\otimes N_\lambda\hsp \dd \omega_\lambda,
\end{align}
where
\begin{align}
N_\lambda &:=\sum_{\mu\nu}\tr(\sigma_\lambda P_{\mu\nu}) N_{\mu\nu}(t,\tau)\nonumber \\ &=\sum_{\mu\nu ij} \tr(\sigma_\lambda P_{\mu\nu})\hsp\ket{i}\bra{j}\otimes\tr_{\rm Env}[\mathcal{U}_{t+\tau|t}(\ket{i}\bra{j}\otimes P_{\mu\nu})] \nonumber \\ &=\sum_{ij} \ket{i}\bra{j}\otimes\tr_{\rm Env}\big[\mathcal{U}_{t+\tau|t}(\ket{i}\bra{j}\otimes \sigma_\lambda^\TT)\big],
\end{align}
which satisfy $N_\lambda\geqcurved 0$ and $\tr_{\rm C}\hsp N_\lambda=\mathrm{id}_{\rm B}$ since, for all $\lambda$, $\sigma_\lambda^\TT$ is an actual density operator. \\

\noindent (ii) Now suppose the map
\begin{equation}
\mathcal{N}_{\mu\nu}(t,\tau)[\cdot]=\tr_{\rm Env}\circ\mathcal{U}_{t+\tau|t}(\cdot \otimes P_{\mu\nu})
\end{equation}
is non-zero only for $\mu=\nu$. In this case
\begin{align}
W(t,\tau)&=\sum_{\mu}\tr_{\rm Env}[\varrho_{\rm SEnv}(t)(\mathrm{id}\otimes \ket{\mu}\bra{\mu})]\otimes N_{\mu\mu}(t,\tau) \nonumber \\ &=\sum_\mu p_\mu\hsp  \varrho_\mu\otimes N_\mu,
\end{align}
where $N_\mu\equiv N_{\mu\mu}$, $N_\mu=\sum_{ij}\tr_{\rm Env}\circ\mathcal{U}_{t+\tau|t}(\ket{i}\bra{j}\otimes \ket{\mu}\bra{\mu})$ and $p_\mu\hsp\varrho_\mu:=\tr_{\rm Env}[\varrho_{\rm SEnv}(t)(\mathrm{id}\otimes \ket{\mu}\bra{\mu})]$, $\tr\varrho_\mu=1$. 
\end{proof}

\section*{Appendix C. Proof of Observation 2}\hypertarget{Proof_Obs2}{ }

\begin{obsapp}
If $W\in\mathbf{CM}$ then
\begin{equation}\label{eq:ER-SM}
\mathcal{E}(W):=\max_{\Theta\in\mathbf{ER}} \tr(\Theta W)\leq 1. 
\end{equation}
Given $\Theta\in\mathbf{ER}$ and $Z_0$ positive on TPMs [i.e., $\tr(Z_0 W)\geq 0$ for all $W\in\mathbf{TPM}$], then
\begin{equation}\label{eq:QM-Wit-SM}
Z(\Theta,Z_0)=Z_0+\kappa(\Theta,Z_0)\left(\Theta-\frac{\mathrm{id}}{\dim(\mathcal{H}_{\rm A}\otimes\mathcal{H}_{\rm B})}\right)
\end{equation}
is a QM witness, with
\begin{equation}\label{eq:kappa-SM}
{\kappa(\Theta,Z_0)}=\min_{\Omega\in\mathbf{CM}}\frac{\tr(Z_0\Omega)}{(\dim\mathcal{H}_{\rm A})^{-1}-\tr(\Theta \Omega)}.
\end{equation}
Conversely, any QM witness can be refined into the form~\eqref{eq:QM-Wit-SM}. 
\end{obsapp}
\begin{proof}
$W\in\mathbf{CM}$ means that $W=\int \hsp\varrho_{\lambda}\otimes N_{\lambda}\hsp \dd \omega_{\lambda}$ for some probability measure $\omega$, density operators $\varrho_{\lambda}$ and quantum channels $N_\lambda$. Since $\varrho_\lambda\leqcurved \mathrm{id}_{\rm A}$, one has
\begin{align}
\mathcal{E}(W)&=\max_{\Theta\in\mathbf{ER}} \tr(\Theta W)\nonumber \\ &=\max_{\Theta\in\mathbf{ER}}\int  \tr[\Theta(\varrho_\lambda\otimes N_\lambda)]\dd\omega_\lambda \nonumber \\  &\leq \max_{\Theta\in\mathbf{ER}}\int\tr[\Theta(\mathrm{id}_{\rm A}\otimes N_\lambda)]\dd\omega_\lambda \nonumber \\ &\leq \max_{\Theta\in\mathbf{ER}}\int\tr[(\sigma\otimes \mathrm{id}_{\rm C})N_{\lambda}]\dd\omega_\lambda\nonumber \\ &=1,
\end{align}
where the trace-constraint of each $N_{\rm \lambda}$, $\tr_{\rm C}\hsp N_{\lambda}=\mathrm{id}_{\rm B}$, was used.

Now let $Z$ be an arbitrary self-adjoint operator on $ABC$, and define
\begin{equation}
\Theta=\frac{\mathrm{id}}{\dim(\mathcal{H}_{\rm A}\otimes \mathcal{H}_{\rm B})}-\beta \left(Z-\sigma\otimes \tr_{\rm A}\hsp Z\right),
\end{equation}
where $\sigma$ is an \emph{arbitrary density operator} on $\mathcal{H}_{\rm B}$. $\Theta$ satisfies $\tr_{\rm A}\hsp \Theta= (\dim\mathcal{H}_{\rm B})^{-1}\mathrm{id}_{\rm BC}$, and there is an $\varepsilon>0$ such that for $\beta\in[-\varepsilon,\varepsilon]$ one has $\Theta\geqcurved 0$; so that
that $\Theta\in\mathbf{ER}$ for the values of $\beta$ in such an interval. One can thus write
\begin{equation}
Z=Z_0+\kappa\left(\Theta-\frac{\mathrm{id}}{\dim(\mathcal{H}_{\rm A}\otimes\mathcal{H}_{\rm B})}\right),
\end{equation}
where $Z_0=\sigma\otimes \tr_{\rm A}\hsp Z$ and $\kappa=-1/\beta$. If one picks $\sigma$ as the maximally mixed state, then for any $W\in\mathbf{TPM}$
\begin{align}
    \tr(Z_0W)&=(\dim\mathcal{H}_{\rm A})^{-1}\tr[(\mathrm{id}_{\rm A}\otimes \tr_{\rm A}\hsp Z)W] \nonumber \\ &=\tr\left[Z\left(\frac{\mathrm{id}_{\rm A}}{\dim\mathcal{H}_{\rm A}}\otimes \tr_{\rm A}\hsp W\right)\right].
\end{align}
If $Z$ is a QM witness then the right-hand side is always non-negative, since $N:=\tr_{\rm A}\hsp W$ satisfies $N\geqcurved 0$ and $\tr_{\rm C}\hsp N=\mathrm{id}_{\rm B}$ for all $W\in\mathbf{TPM}$. Finally, consider $Z$ as a function of $\kappa$, keeping $Z_0$ and $\Theta$ fixed. Each valid value of $\kappa$ (i.e., such that $Z$ is a QM witness) defines a hyperplane that separates the set $\bf CM$ from the other TPM processes. Thus, the optimal value for $\kappa$ consists of the one for which 
\begin{equation}
\min_{\Omega\in\mathbf{CM}}\tr(Z\Omega)=0,
\end{equation}
whence Eqs.~\eqref{eq:QM-Wit-SM} and \eqref{eq:kappa-SM} follow directly.
\end{proof}

\section*{Appendix D. Proof of Observation 3}\hypertarget{AppendixD}{ }

\begin{obsapp}\hypertarget{observation2}
The maximum average success probability is characterized 
by
\begin{equation}
\label{eq:maximum_sucess}
\max_{\mathcal{S}\in\mathbf{Strat}}
\frac{1}{|\mathcal{X}|}\sum_{x\in\mathcal{X}}\mathrm{Pr}(y=x|x,\mathcal{S})
= \frac{\mathcal{E}(W)}{\dim \mathcal{H}_{\rm A}},
\end{equation}
where $\mathrm{P}(y=x|x,\mathcal{S})$ reads as the probability of correctly decoding input $x$ in a strategy $\mathcal{S}$.
\end{obsapp}

\begin{proof}
In what follows, ``$\star$'' will denote so-called link product \cite{CDP09}. Let $\Theta\in\mathbf{ER}$ and consider
a POVM $\mathcal{X} \ni y \mapsto F_y \in\mathrm{L}(\mathcal{H}_{\rm A}\otimes \mathcal{H}_{\rm A^\prime})$,
aiming to decode the value of $x$. Here, $\mathcal{H}_{\rm A^\prime}\simeq \mathcal{H}_{\rm A}$ and $|\mathcal{X}|<\infty$. Then the mapping $y \mapsto E_y=F_y\star \Theta\in\mathrm{L}(\mathcal{H}_{\rm A^\prime}\otimes \mathcal{H}_{\rm B}\otimes \mathcal{H}_{\rm C})$ satisfies
\begin{align}
\sum_{y\in\mathcal{X}} E_y &=\sum_{y\in\mathcal{X}} F_y\star \Theta \nonumber \\ 
&=\mathrm{id}_{\rm AA^\prime}\star \Theta  \nonumber \\ 
&= \mathrm{id}_{\rm A^\prime}\otimes \tr_{\rm A}\hsp \Theta \nonumber \\ 
&\leq \mathrm{id}_{\rm A^\prime}\otimes \sigma\otimes \mathrm{id}_{\rm C}.
\end{align}
Therefore, if one defined the additional inconclusive tester effect as
\begin{equation}
E_{\emptyset}:=\mathrm{id}_{\rm A^\prime}\otimes \sigma\otimes \mathrm{id}_{\rm C}-\sum_{x\in\mathcal{X}} E_x,
\end{equation}
then the mapping $\{\emptyset\}\cup \mathcal{X}\ni y\mapsto E_y$ defines a tester. 

Next, consider $F_y$ as the Bell state measurement over $\mathcal{H}_{\rm A}\otimes \mathcal{H}_{\rm A^\prime}$ [therefore, $|\mathcal{X}|=(\dim\mathcal{H}_{\rm A})^2$]. From the CJ isomorphism, each $F_y$ is uniquely associated to a unitary transformation, say $U_x$. Denote by $G_x$ the CJ operator corresponding to the unitary transformation $U_x^{-1}$. Since $U_x U_x^{-1}$ is the identity, one has that with suitable relabeling of the Hilbert spaces,
\begin{align}
\tr(\Theta W)&=(\dim\mathcal{H}_{\rm A})\tr[(F_x\star G_x)\Theta W] &\forall x\in\mathcal{X}\nonumber \\ &=(\dim\mathcal{H}_{\rm A})\tr[(F_x\star \Theta)(W\star G_x)] &\forall x\in\mathcal{X} \nonumber \\  &=\frac{1}{\dim\mathcal{H}_{\rm A}}\sum_{x\in\mathcal{X}}\tr(E_x W_x) \nonumber \\ &=\frac{1}{|\mathcal{X}|^{1/2}}\sum_{x\in\mathcal{X}}\mathrm{P}(y=x|x;\mathcal{S}),
\end{align}
where the strategy $\mathcal{S}$ is defined by the unitary channel encodings $\mathcal{U}_y$ and the tester defined via link product of the CJ operator of the Bell state measurement effects, $F_y$, and the ER $\Theta$. Therefore
\begin{equation}\label{eq:ineq}
\max_{\Theta\in\mathbf{ER}}\tr(\Theta W)\leq \max_{\mathcal{S}\in\mathbf{Strat}}|\mathcal{X}|^{-1/2}\sum_{x\in\mathcal{X}}\mathrm{Pr}(y=x|x;\mathcal{S}).
\end{equation}
and especially that equality in Eq.~(\ref{eq:maximum_sucess})
can be reached.

To prove the reverse inequality, and thus establish the equality, consider the optimal strategy $\mathcal{S}_{\rm opt}$ as defined by a tester $\mathcal{X}\ni x\mapsto E_x^{\rm opt}\in\mathrm{L}(\mathcal{H}_{\rm A^\prime}\otimes \mathcal{H}_{\rm B}\otimes \mathcal{H}_{\rm C})$ and encoding channels $\mathcal{X}\ni x\mapsto M_x^{\rm opt}\in\mathrm{L}(\mathcal{H}_{\rm A}\otimes \mathcal{H}_{\rm A^\prime})$, where $\tr_{\rm A^\prime}\hsp M_x^{\rm opt}=\mathrm{id}_{\rm A}$. The additional assumptions to state the inequality is that the dimensionality of the output Hilbert space of each channel $M_x^{\rm opt}$ cannot exceed the dimensionality of the input and that they represent unital channels, $\tr_{\rm A}\hsp M_x^{\rm opt}=\mathrm{id}_{\rm A^\prime}$. Without loss of generality, assume $\mathcal{H}_{\rm A}\simeq \mathcal{H}_{\rm A^\prime}$. With the help of the link product, one can write
\begin{align}
\sum_{x\in\mathcal{X}}\mathrm{P}(y=x|x,\mathcal{S}_{\rm opt})&=\sum_{x\in\mathcal{X}}\tr\big[E_x^{\rm opt}\big(W\star M_x^{\rm opt}\big)\big] \nonumber \\ &=\sum_{x\in\mathcal{X}}\tr\big[\big(M_x^{\rm opt}\star E_x^{\rm opt}\big)W\big].
\end{align}
Finally, conclude that
\begin{equation}
\Theta_{\rm opt}:=\frac{1}{|\mathcal{X}|^{1/2}}\sum_{x\in\mathcal{X}}M_x^{\rm opt}\star E_x^{\rm opt}\in \mathbf{ER},
\end{equation}
since $\Theta\geq 0$ and
\begin{align}\label{eq:ineqproof}
\tr_{\rm A}\hsp \Theta &=\frac{1}{|\mathcal{X}|^{1/2}}\sum_{x\in\mathcal{X}}\mathrm{id}_{\rm A} \star (M_x^{\rm opt}\star E_x^{\rm opt}) \nonumber \\ &=\frac{\dim\mathcal{H}_{\rm A}}{|\mathcal{X}|^{1/2}}\sigma\otimes \mathrm{id}_{\rm C} \nonumber \\ &\leq \sigma\otimes \mathrm{id}_{\rm C},
\end{align}
as long as $|\mathcal{X}|\geq (\dim\mathcal{H}_{\rm A})^2$, where
\begin{equation}
\sigma=\frac{1}{\dim\mathcal{H}_{\rm C}}\tr_{\rm AC}\sum_{x\in\mathcal{X}}E_x.
\end{equation}
Moreover, the fact that $M_{x}^{\rm opt}$ is unital was used, namely, $\tr_{\rm A}\hsp M_x^{\rm opt}=\mathrm{id}_{\rm A^\prime}$. Consequently,
\begin{equation}
\frac{1}{|\mathcal{X}|^{1/2}}\sum_{x\in\mathcal{X}}\mathrm{Pr}(y=x|x;\mathcal{S}_{\rm opt})\leq \max_{\Theta\in\mathbf{ER}}\tr(\Theta W),
\end{equation}
which, together with Ineq.~\eqref{eq:ineq}, leads to the equality in Eq.~\eqref{eq:maximum_sucess}.
\end{proof}

\section*{Appendix E. Proof of Observation 4}\hypertarget{AppendixE}{ }

\begin{obsapp}\hypertarget{obs4}{ }
The action of a QM witness in the Heisenberg picture can be always be written as
\begin{equation}\label{eq:memory_HeisenbergApp}
\tr[Z W(t,\tau)]=\sum_{ij}d_{ij}(Z)\hsp\mathfrak{g}^{\textrm{\tiny$(2)$}}(E_i,F_j;t,\tau),
\end{equation}
where $d_{ij}(Z)$ are real coefficients and the operators $E_i$ and $F_j$ are fixed, i.e., at the initial time they do not depend on the process nor the witness.
\end{obsapp}
\begin{proof}
In the CJ picture, the action of a QM witness $Z$ in a TPM process $W(t_1,t_2)$ reads
\begin{equation}\label{eq:Mt1t2}
m_{Z}(t,\tau)=\tr[ZW(t,\tau)].
\end{equation}
Let $M_i\in\mathrm{L}(\mathcal{H}_{\rm A}\otimes \mathcal{H}_{\rm B})$ and $F_j\in\mathrm{L}(\mathcal{H}_{\rm B})$ be positive semidefinite operators spanning their corresponding operator spaces, i.e., frames. Thus,
\begin{equation}
Z=\sum_{ij} d_{ij}(Z)M_i\otimes F_j,
\end{equation}
where $d_{ij}(Z)=\tr[(\tilde{M}_i\otimes \tilde{F}_j)Z]$, with $\tilde{M}_i$ (resp. $\tilde{F}_j$) denoting the dual of $M_i$ (resp. $F_j$). Therefore,
\begin{equation}
m_Z(t,\tau)=\sum_{ij}d_{ij}(Z)\tr[(M_i\otimes F_j)W(t,\tau)].
\end{equation}
Without loss of generality, consider the operators $F_j$ as defining an IC-POVM and $\tr_{\rm B}\hsp M_i\leq \mathrm{id}_{\rm A}$. From the channel-state duality, each operator $M_i$ can be mapped into a CP map,
\begin{equation}
M_i\mapsto \sum_{\gamma_i} E_{\gamma_i}\cdot E_{\gamma_i}.
\end{equation}
Plugging this into Eq.~\eqref{eq:Mt1t2}, moving to the Heisenberg picture (see \hyperlink{Heisenberg_to_CJ}{Appendix A}), and redefining the indices one gets Eq.~\eqref{eq:memory_HeisenbergApp}.
\end{proof}

\section*{Appendix F. Heat flow and the role of the spectral density in the spin-Boson model}\hypertarget{appendixF}{ }

In the context of the spin-boson model, we derive two main results, corresponding to Eqs.~(19) and (20) in the main text. Here, we discuss this results in more details, by showing their derivations as well as additional illustrative examples. The key point of our analysis is to find a suitable QM witness in the limiting case where the process matrix can be obtained analytically, then extend the analysis for the general case.

To start off, consider the spin-boson Hamiltonian in the frame rotating together with the free Hamiltonian,
\begin{equation}
H_{\rm RF}(t)=\tilde{B}(t)^\dagger \sigma+\sigma^\dagger \tilde{B}(t),
\end{equation}
where
\begin{equation}
\tilde{B}(t)=\sum_{\lambda}\ee^{-\ii (\omega_\lambda-\omega_0) t}\hsp \hbar g_\lambda \hsp b_\lambda.
\end{equation}
Since the total number of excitations is conserved, i.e., the ``number operator''
\begin{equation}
N=\sigma^\dagger \sigma+\sum_{\lambda}b_\lambda^\dagger b_{\lambda}^{ }
\end{equation}
commutes with the Hamiltonian, then, if the system is initialized in the state
\begin{equation}\label{eq:initial_state}
\ket{\psi(0)}=\ket{e}\otimes \ket{\mathrm{vac}},
\end{equation} 
it will evolve to
\begin{equation}\label{eq:single-exc}
\ket{\psi(t)}=q(t) \ket{e}\otimes \ket{\mathrm{vac}}+\sum_{\lambda} \alpha_\lambda(t)\ket{g}\otimes b_\lambda^\dagger\ket{\mathrm{vac}},
\end{equation}
where $q$ and $\alpha_\lambda$ are probability amplitudes. That is, the wave function of the system is trapped in the ``single excitation subspace'', which significantly simplifies the analysis. Plugging Eq.~\eqref{eq:single-exc} into the Schr\"{o}dinger equation,
\begin{equation}
\ii \hbar \frac{\dd}{\dd t}\ket{\psi(t)}=H_{\rm RF}(t)\ket{\psi(t)},
\end{equation}
we obtain
\begin{subeqnarray}
\dot{q}(t)&=&-\ii \sum_{\lambda} g_\lambda\hsp \ee^{-\ii (\omega_\lambda- \omega_0)t}\alpha_\lambda(t), \\ 
\dot{\alpha}_{\lambda}(t)&=&-\ii g_\lambda\hsp  \ee^{\ii (\omega_\lambda- \omega_0)t}q(t).
\end{subeqnarray}
Integrating the second equation and plugging the solution into the first we obtain the following integral equation
\begin{equation}\label{eq:convolution}
\dot{q}(t)=-\int_0^t f(t-s)q(s)\dd s,
\end{equation}
where the integration kernel is given by the bath two-time correlation function,
\begin{equation}
f(u)=\tr[B(u)B\hsp \varrho_{\rm Env}]\ee^{\ii  \omega_0 u}.
\end{equation}
The Fourier transform of the bath two-time correlation function defines the bath spectral density $J(\omega)$, 
\begin{equation}
f(u)=\int J(\omega)\ee^{\ii ( \omega_0-\omega)u}\hsp\dd\omega.
\end{equation}
Since the right-hand side of Eq.~\eqref{eq:convolution} is a convolution, it can be solved by means of the Laplace ($\rm Lap$) transform,
\begin{equation}\label{eq:Laplace}
\mathrm{Lap}[q(t)](w)=\frac{1}{w+\mathrm{Lap}[f(t)](w)},
\end{equation}
where the initial condition $q(0)=1$ was used [Eq.~\eqref{eq:initial_state}]. In the single-excitation subspace, the TPM process matrix $W(t,\tau)$ can be computed in terms of the amplitudes $q$ and $\alpha_k$ by using our \hyperlink{Lemma2}{Lemma 2}. 

\begin{figure*}
    \centering
    (a)\includegraphics[width=0.45\linewidth]{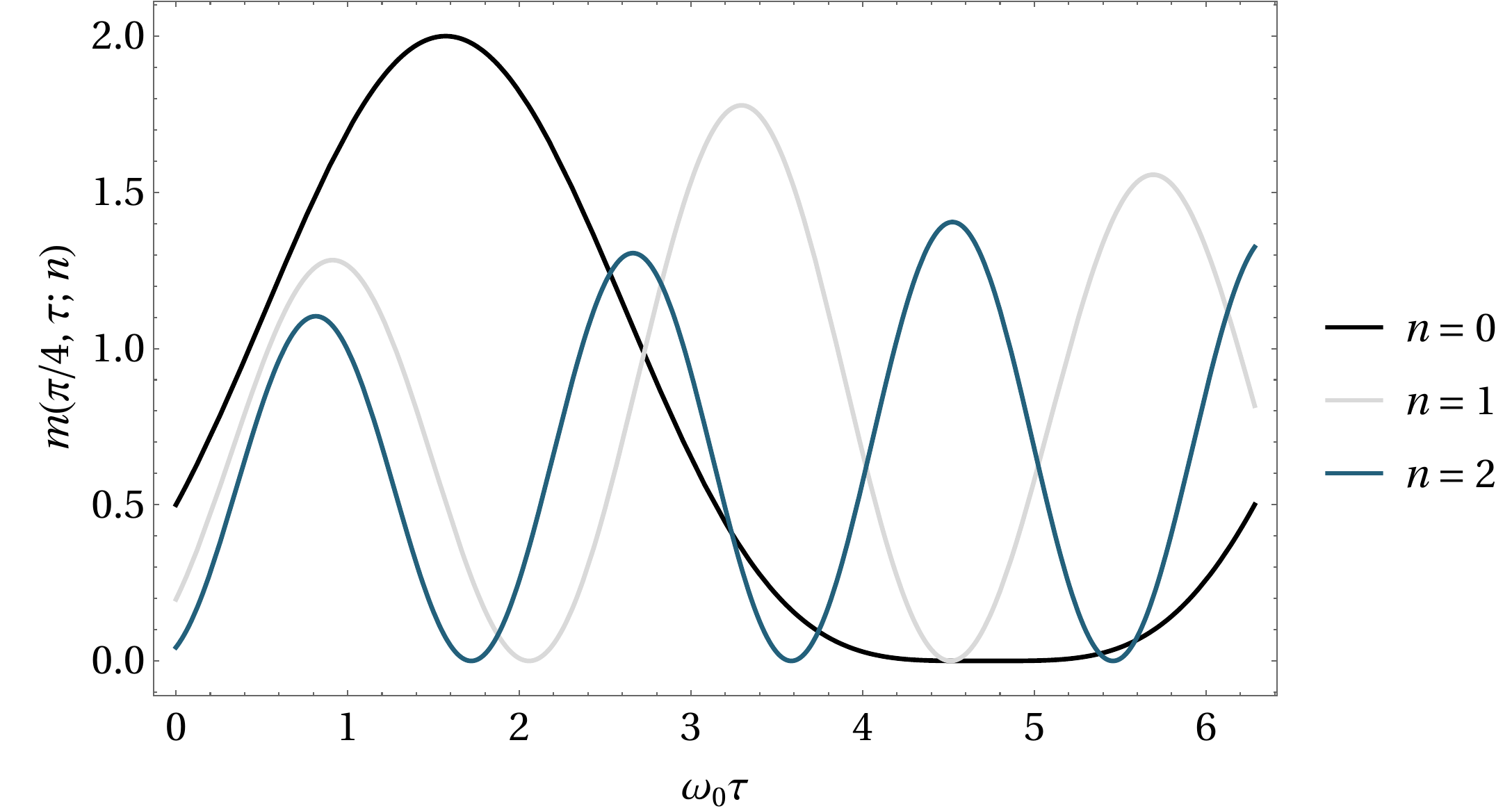}
    (b)\includegraphics[width=0.49\linewidth]{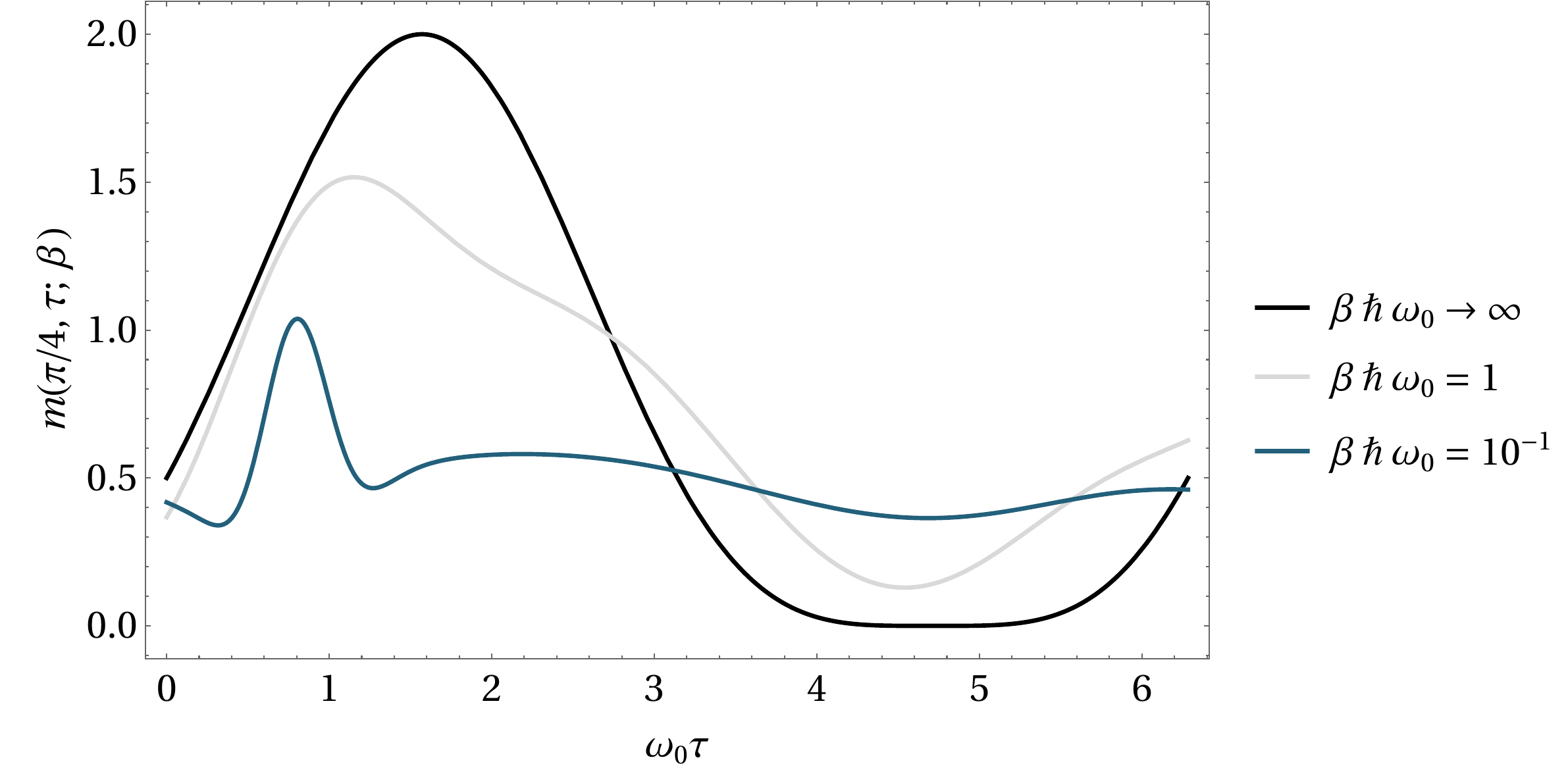}
    \caption{Quantum memory detection [$m(\pi/4,\tau)>1$] for Fock (a) and thermal states (b) in the spin-boson model.}
    \label{fig:fock-thermal}
\end{figure*}

To find a QM witness, we consider the simplest case of a single-mode resonant cavity, wherein
\begin{equation}\label{eq:single-mode}
J(\omega)=g^2 \delta(\omega- \omega_0).
\end{equation}
This corresponds to the resonant Jaynes-Cummings model in the rotating wave approximation.
We run our SDP method for this case,
\begin{align}
m(t,\tau)&=\mathcal{E}(W(t,\tau))\nonumber \\ &=\max_{\Theta\in\mathbf{ER}}\tr[W(t,\tau)\Theta].
\end{align}
By doing so we realize that the solution of this optimization problem (up to numerical precision) yields
\begin{equation}\label{eq:optimalER}
\Theta_*=2\hsp\ket{g}\bra{g}\otimes \ket{\Psi^-}\bra{\Psi^-},
\end{equation}
where $\ket{g}\in\mathcal{H}_{\rm B}$ and $\ket{\Psi^-}\in\mathcal{H}_{\rm A}\otimes \mathcal{H}_{\rm C}$ is the singlet state,  
\begin{equation}
\ket{\Psi^-}=\frac{1}{\sqrt{2}}(\ket{g}\otimes \ket{e}-\ket{e}\otimes \ket{g}).
\end{equation}
Clearly, $\Theta_*\in\mathbf{ER}$ since $\Theta_*\geq 0$ and
\begin{equation}
\tr_{\rm A}\hsp \Theta_*=\mathrm{id}_{\rm A}\otimes \ket{g}\bra{g}.
\end{equation}
Moreover, notice that
\begin{align}
\max_{W\in\mathbf{CM}}\tr(\Theta_* W)&=1,
\end{align}
and therefore the QM witness associated to $\Theta_*$ is already optimal (i.e., no more SDPs need to be computed). 
Plugging $\Theta_\star$ into $\tr[W(t,\tau)\Theta_\star]$,  we obtain
\begin{equation}\label{eq:witnessSD}
m(t,\tau)=\left|q(t)+\int_0^t\int_0^\tau f(t+\tau-r-s)q(r)q(s)\dd r \dd s\right|^2.
\end{equation}
This expression is one of our main results within this model [Eq.~(20) in the main text].

\begin{rem}
Even though we have used numerical optimization in the single-mode case [Eq.~\eqref{eq:single-mode}] to arrive at the form of $\Theta_*$, once we have it, Eq.~\eqref{eq:witnessSD} is valid for any spectral density. Of course, it might be the case that $\Theta_\star$ does not necessarily lead to an optimal QM witness outside the resonant single-mode case, but the violation of the inequality we derived, namely $m(t,\tau) \leq 1$,
is still a definite signature of non-classicality regardless of the spectral density.
\end{rem}

In the main text, we explored Eq.~\eqref{eq:witnessSD} for the case of a resonant Lorentzian spectral density,
\begin{equation}
J(\omega)=\frac{\lambda^2 \gamma_0}{2\pi[( \omega_0-\omega)^2+\lambda^2]},
\end{equation}
and initial state as in Eq.~\eqref{eq:initial_state}. To show that our results are \emph{not} limited to this particular example, we consider two additional examples.

\begin{example}[\sc Different initial states]
Consider again the single-mode resonant spectral density, $J(\omega)=g^2\delta(\omega- \omega_0)$, but with a more general initial state,
\begin{equation}
\varrho_{\rm SEnv}(0)=\ket{e}\bra{e}\otimes \varrho_{\rm Env},
\end{equation}
such that $\varrho_{\rm Env}$ is diagonal in the Fock basis of $H_{\rm Env}=\hbar \omega_0b^\dagger b$. 
As the total Hamiltonian is easily diagonalizable for a single bath mode, we do not need to restrict ourselves to the single excitation space. 
For an initial Fock state $\ket{n}$ one has
\begin{align}
m(t,\tau;n)&=\Big[\sin \left(\sqrt{n+1} gt\right) \sin \left(\sqrt{n+1} g\tau\right)\nonumber \\ &\quad+\cos \left(\sqrt{n+1} gt\right) \cos \left(\sqrt{n}g\tau\right)\Big]^2.
\end{align}
For a thermal state, one takes the Gibbs mixture over $n$. Fig.~\ref{fig:fock-thermal} shows exemplary values of $m$ for (a) various Fock states and (b) various inverse temperatures $\beta = 1/k_B T$. Quantum memory can be detected for states of low $n$ or temperature.
\end{example}

\begin{example}[\sc Ohmic spectral density]
As a more involved example, consider an Ohmic spectral density with hard cutoff,
\begin{equation}
J(\omega)=\eta \omega\hsp  \Theta_{\rm Heaviside}(\omega_{\rm C}-\omega).
\end{equation}
The cutoff frequency $\omega_{\rm C}$ determines how structured the bath is (the bigger it is the more ``Markovian'' the bath). The two-time correlation function equals
\begin{equation}
f(t)=\frac{\eta  \left(\ii t \omega_{\rm C}-\ee^{\ii \omega_{\rm C} t}+1\right) \ee^{\ii t ( \omega_0-\omega_{\rm C})}}{t^2}.
\end{equation}
Its Laplace transform reads
\begin{equation}
\mathrm{Lap}[f(t)](w)=-\ii \eta  \omega_{\rm C}-\eta (w-\ii  \omega_0) \log \left(\frac{w-\ii  \omega_0}{w-\ii ( \omega_0-\omega_{\rm C})}\right), 
\end{equation}
which determines the solution for $q(t)$ via \eqref{eq:Laplace}.
We numerically invert the Laplace transform to compute $m(t,\tau)$, which we then maximize over times $t$ and $\tau$ to obtain the strongest QM signature $m(t,\tau)$ for different cutoffs. The results are plotted in Fig.~\ref{fig:cutoff}. QM is detectable for $\omega_C \lesssim 3 \omega_0$ in this example.
\end{example}

\begin{figure}
    \centering
    \includegraphics[width=0.8\linewidth]{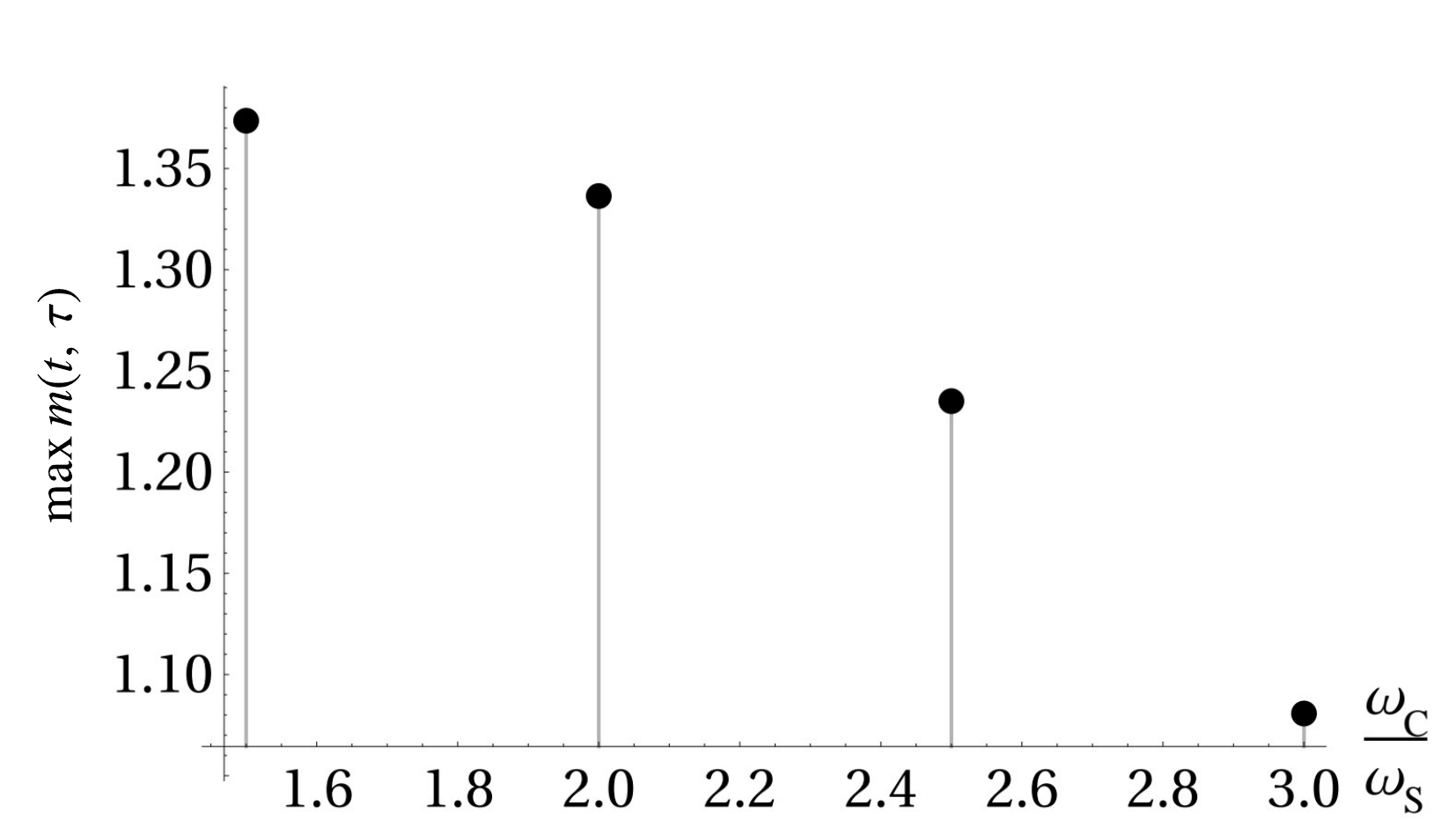}
    \caption{Numerical estimation of the cutoff frequency after which quantum memory is no longer detectable ($\omega_{\rm C}\approx 3 \omega_0$) via $m(t,\tau)$ (numerically maximised over different times) in the Ohmic spectral density with hard cutoff.}
    \label{fig:cutoff}
\end{figure}

In the main text, we present an inequality [Eq. (18) there] to detect quantum memory in a generic spin-boson model (under rotating-wave approximation in the Hamiltonian), which is independent on both spectral density and initial system-environment state. Such inequality is
\begin{align}\label{eq:IneqDJC}
2\hbar \omega_0\mathrm{Re}\Big\langle&{N}(t)^{-\frac{1}{2}}\sin\big({N}(t)^\frac{1}{2}\tau\big){B}(t)^\dagger {\sigma}(t)\cos\big({N}(t)^\frac{1}{2}\tau\big)\Big\rangle\nonumber \\ &+\Big\langle\cos\Big(2{N}(t)^{\frac{1}{2}}\tau\Big){H}_{\rm S}(t)\Big\rangle\leq \frac{\hbar \omega_0}{2},
\end{align}
where $N(t)\equiv B(t)^\dagger B(t)$. The derivation of this bound is length as it involves applications of the previous results together with symbolic calculations that was done with the help of \textsc{Mathematica} that will be omitted. In what follows, we outline the steps of the proof.

The first step is to fix a suitable frame for the operator spaces. Let $\mathbf{M}_{\rm AB}=\{M_a\in \mathrm{L}(\mathcal{H}_{\rm A}\otimes \mathcal{H}_{\rm B}):a\in\mathcal{A}\}$ and $\mathbf{F}_{\rm C}=\{F_b \in\mathrm{L}(\mathcal{H}_{\rm C}):b\in\mathcal{B}\}$ be frames for the corresponding operator spaces, i.e., finite set of operators such that $\mathrm{Span}\hsp \mathbf{M}_{\rm AB}=\mathrm{L}(\mathcal{H}_{\rm A}\otimes \mathcal{H}_{\rm B})$ and $\mathrm{Span}\hsp \mathbf{F}_{\rm C}=\mathrm{L}(\mathcal{H}_{\rm C})$. Since $\dim(\mathcal{H}_{\rm A}\otimes\mathcal{H}_{\rm B})=2\dim\mathcal{H}_{\rm C}=4$, SIC-POVMs are suitable frames since their dual set are easily computable. 

A set of rank-1 projectors $\{\Pi_x:0\leq x\leq d^2-1 \}\subset\mathrm{L}(\mathcal{H})$, $d=\dim\mathcal{H}$, defines a SIC-POVM if
\begin{equation}
\tr\big(\Pi_\mu^{(d)} \Pi_\nu^{(d)}\big)=\frac{d\delta_{\mu\nu}+1}{d+1},
\end{equation}
where the POVM effects are defined via $P_x^{(d)}=\Pi_x^{(d)}/d$. The dual elements are
\begin{equation}
\tilde{P}_x^{(d)}=d(d+1)P_x^{(d)}-\mathrm{id},
\end{equation}
which satisfies $\tr\big(\tilde{P}_\mu^{(d)}P_\nu^{(d)}\big)=\delta_{\mu\nu}$. Explicit expressions for SIC-POVMs of dimension 2 and 4 can be easily found in many references, e.g., \cite{Bengtsson2024,Appleby2025}. Define
\begin{equation}
d_{ij}(\Theta_*)=\tr\big[\big(\tilde{P}_i^{(4)}\otimes \tilde{P}_j^{(2)}\big)\Theta_*\big].
\end{equation}
In this representation one has
\begin{equation}
\Theta_*=\sum_{ij}d_{ij}(\Theta_*)\hsp P_i^{(4)}\otimes P_j^{(2)}.
\end{equation}
Next, $P_i^{(4)}$ is mapped into a CP map via channel-state duality, $P_i^{(4)}\mapsto E_i\cdot E_i^\dagger$, where
\begin{equation}
E_a=\sum_{ij}\bra{\Pi_a^{(4)}}(\ket{i}\otimes \ket{j})\hsp\ket{i}\bra{j},
\end{equation}
$\Pi_a^{(4)}=\ket{\Pi_a^{(4)}}\bra{\Pi_a^{(4)}}$. The $2^2\times 4^2=64$ real numbers $d_{ij}(\Theta_*)$ as well as the $16$ ($2\times 2$)-matrices $E_a$ can be easily computed with \textsc{Mathematica} (or any other software for symbolic calculation).

From \hyperlink{obs4}{Observation 4}, the action of $\Theta_*$ upon a TPM process matrix $W(t,\tau)$ in the Heisenberg picture reads
\begin{equation}
\tr[\Theta_* W(t,\tau)]=\sum_{ij}d_{ij}(\Theta_*)\hsp\mathfrak{g}^{\textrm{\tiny$(2)$}}(E_i,F_j;t,\tau).
\end{equation}
Let $t_1=t_0+t$ and $t_2=t_1+\tau$ and write $m_{Z_{\star}}$ as a formal power series of $\tau$,
\begin{align}\label{eq:power_series}
m_{\Theta_*}(t,\tau)&=\sum_{n\geq 0}\frac{\tau^n}{n!}\frac{\partial^{n}}{\partial T^{n}}m_{\Theta_*}(t,T)\Big|_{T=0}\nonumber \\ &=\sum_{ab}\sum_{n\geq 0}\frac{\tau^n}{n!}d_{ij}(\Theta_*)\nonumber \\ &\times\Big\langle\hat{E}_a(t_1)^\dagger \frac{\partial^{n}}{\partial T^{n}} \hat{F}_b(T)\Big|_{T=0} \hat{E}_a(t_1)\Big\rangle.
\end{align}

The next step is to move to a more convenient frame. In the frame rotating together with the free Hamiltonian (Dirac's interaction picture), one has
\begin{equation}
H_{\rm RF}(t)=\big(\tilde{B}(t)^\dagger \sigma+\sigma^\dagger \tilde{B}(t)\big),
\end{equation}
where
\begin{equation}
\tilde{B}(t)=\sum_{\lambda}\hbar g_\lambda\ee^{-\ii (\omega_\lambda-\omega_0) t} b_\lambda.
\end{equation}
The evolution of a system operator $X$ in the Heisenberg picture in the rotating frame is given by
\begin{align}
\dot{\hat{X}}(t)&=-\frac{1}{\ii\hbar}[\hat{H}_{\star-RF}(t),\hat{X}(t)] \nonumber \\ &=[\hat{L}(t)^\dagger \hat{B}(t),\hat{X}(t)]+[\hat{B}(t)^\dagger \hat{L}(t),\hat{X}(t)]\nonumber \\ &\equiv \mathfrak{H}[\hat{X}(t)]
\end{align}
The derivatives in Eq.~\eqref{eq:power_series} can thus be recursively computed via
\begin{equation}
\frac{\partial^{n}}{\partial T^n}\hat{X}(T)\Big|_{T=t_0+t}=\mathfrak{H}\left[\frac{\partial^{n-1}}{\partial T^{n-1}}\hat{X}(T)\Big|_{T=t_0+t}\right].
\end{equation}
By expanding in different orders of $\tau$, we conjectured the form of the inequality given in Eq.~\eqref{eq:IneqDJC}, based on its consistency with the Taylor expansion. Finally, we verified this conjecture by treating it as an induction hypothesis, which was confirmed using {\sc Mathematica}.


\begin{thebibliography}{89}%
\makeatletter
\providecommand \@ifxundefined [1]{%
 \@ifx{#1\undefined}
}%
\providecommand \@ifnum [1]{%
 \ifnum #1\expandafter \@firstoftwo
 \else \expandafter \@secondoftwo
 \fi
}%
\providecommand \@ifx [1]{%
 \ifx #1\expandafter \@firstoftwo
 \else \expandafter \@secondoftwo
 \fi
}%
\providecommand \natexlab [1]{#1}%
\providecommand \enquote  [1]{``#1''}%
\providecommand \bibnamefont  [1]{#1}%
\providecommand \bibfnamefont [1]{#1}%
\providecommand \citenamefont [1]{#1}%
\providecommand \href@noop [0]{\@secondoftwo}%
\providecommand \href [0]{\begingroup \@sanitize@url \@href}%
\providecommand \@href[1]{\@@startlink{#1}\@@href}%
\providecommand \@@href[1]{\endgroup#1\@@endlink}%
\providecommand \@sanitize@url [0]{\catcode `\\12\catcode `\$12\catcode `\&12\catcode `\#12\catcode `\^12\catcode `\_12\catcode `\%12\relax}%
\providecommand \@@startlink[1]{}%
\providecommand \@@endlink[0]{}%
\providecommand \url  [0]{\begingroup\@sanitize@url \@url }%
\providecommand \@url [1]{\endgroup\@href {#1}{\urlprefix }}%
\providecommand \urlprefix  [0]{URL }%
\providecommand \Eprint [0]{\href }%
\providecommand \doibase [0]{https://doi.org/}%
\providecommand \selectlanguage [0]{\@gobble}%
\providecommand \bibinfo  [0]{\@secondoftwo}%
\providecommand \bibfield  [0]{\@secondoftwo}%
\providecommand \translation [1]{[#1]}%
\providecommand \BibitemOpen [0]{}%
\providecommand \bibitemStop [0]{}%
\providecommand \bibitemNoStop [0]{.\EOS\space}%
\providecommand \EOS [0]{\spacefactor3000\relax}%
\providecommand \BibitemShut  [1]{\csname bibitem#1\endcsname}%
\let\auto@bib@innerbib\@empty
%</preamble>
\bibitem [{\citenamefont {Gardiner}(1985)}]{Gardiner1985}%
  \BibitemOpen
  \bibfield  {author} {\bibinfo {author} {\bibfnamefont {C.~W.}\ \bibnamefont {Gardiner}},\ }\href {https://doi.org/10.1007/978-3-662-02452-2} {\emph {\bibinfo {title} {Handbook of Stochastic Methods for Physics, Chemistry and the Natural Sciences}}}\ (\bibinfo  {publisher} {Springer Berlin Heidelberg},\ \bibinfo {year} {1985})\BibitemShut {NoStop}%
\bibitem [{\citenamefont {Mandelbrot}\ and\ \citenamefont {Van~Ness}(1968)}]{Mandelbrot1968}%
  \BibitemOpen
  \bibfield  {author} {\bibinfo {author} {\bibfnamefont {B.~B.}\ \bibnamefont {Mandelbrot}}\ and\ \bibinfo {author} {\bibfnamefont {J.~W.}\ \bibnamefont {Van~Ness}},\ }\bibfield  {title} {\bibinfo {title} {Fractional brownian motions, fractional noises and applications},\ }\href {https://doi.org/10.1137/1010093} {\bibfield  {journal} {\bibinfo  {journal} {SIAM Review}\ }\textbf {\bibinfo {volume} {10}},\ \bibinfo {pages} {422–437} (\bibinfo {year} {1968})}\BibitemShut {NoStop}%
\bibitem [{\citenamefont {Vijay}\ \emph {et~al.}(2012)\citenamefont {Vijay}, \citenamefont {Macklin}, \citenamefont {Slichter}, \citenamefont {Weber}, \citenamefont {Murch}, \citenamefont {Naik}, \citenamefont {Korotkov},\ and\ \citenamefont {Siddiqi}}]{Vijay2012}%
  \BibitemOpen
  \bibfield  {author} {\bibinfo {author} {\bibfnamefont {R.}~\bibnamefont {Vijay}}, \bibinfo {author} {\bibfnamefont {C.}~\bibnamefont {Macklin}}, \bibinfo {author} {\bibfnamefont {D.~H.}\ \bibnamefont {Slichter}}, \bibinfo {author} {\bibfnamefont {S.~J.}\ \bibnamefont {Weber}}, \bibinfo {author} {\bibfnamefont {K.~W.}\ \bibnamefont {Murch}}, \bibinfo {author} {\bibfnamefont {R.}~\bibnamefont {Naik}}, \bibinfo {author} {\bibfnamefont {A.~N.}\ \bibnamefont {Korotkov}},\ and\ \bibinfo {author} {\bibfnamefont {I.}~\bibnamefont {Siddiqi}},\ }\bibfield  {title} {\bibinfo {title} {Stabilizing {R}abi oscillations in a superconducting qubit using quantum feedback},\ }\href {https://doi.org/10.1038/nature11505} {\bibfield  {journal} {\bibinfo  {journal} {Nature}\ }\textbf {\bibinfo {volume} {490}},\ \bibinfo {pages} {77–80} (\bibinfo {year} {2012})}\BibitemShut {NoStop}%
\bibitem [{\citenamefont {Alhambra}\ \emph {et~al.}(2020)\citenamefont {Alhambra}, \citenamefont {Riddell},\ and\ \citenamefont {Garcia-Pintos}}]{Alhambra2020}%
  \BibitemOpen
  \bibfield  {author} {\bibinfo {author} {\bibfnamefont {A.~M.}\ \bibnamefont {Alhambra}}, \bibinfo {author} {\bibfnamefont {J.}~\bibnamefont {Riddell}},\ and\ \bibinfo {author} {\bibfnamefont {L.~P.}\ \bibnamefont {Garcia-Pintos}},\ }\bibfield  {title} {\bibinfo {title} {Time evolution of correlation functions in quantum many-body systems},\ }\href {http://dx.doi.org/10.1103/PhysRevLett.124.110605} {\bibfield  {journal} {\bibinfo  {journal} {Physical Review Letters}\ }\textbf {\bibinfo {volume} {124}} (\bibinfo {year} {2020})}\BibitemShut {NoStop}%
\bibitem [{\citenamefont {Dowling}\ \emph {et~al.}(2023{\natexlab{a}})\citenamefont {Dowling}, \citenamefont {Figueroa-Romero}, \citenamefont {Pollock}, \citenamefont {Strasberg},\ and\ \citenamefont {Modi}}]{Dowling2023a}%
  \BibitemOpen
  \bibfield  {author} {\bibinfo {author} {\bibfnamefont {N.}~\bibnamefont {Dowling}}, \bibinfo {author} {\bibfnamefont {P.}~\bibnamefont {Figueroa-Romero}}, \bibinfo {author} {\bibfnamefont {F.~A.}\ \bibnamefont {Pollock}}, \bibinfo {author} {\bibfnamefont {P.}~\bibnamefont {Strasberg}},\ and\ \bibinfo {author} {\bibfnamefont {K.}~\bibnamefont {Modi}},\ }\bibfield  {title} {\bibinfo {title} {Relaxation of multitime statistics in quantum systems},\ }\href {http://dx.doi.org/10.22331/q-2023-06-01-1027} {\bibfield  {journal} {\bibinfo  {journal} {Quantum}\ }\textbf {\bibinfo {volume} {7}},\ \bibinfo {pages} {1027} (\bibinfo {year} {2023}{\natexlab{a}})}\BibitemShut {NoStop}%
\bibitem [{\citenamefont {Dowling}\ \emph {et~al.}(2023{\natexlab{b}})\citenamefont {Dowling}, \citenamefont {Figueroa-Romero}, \citenamefont {Pollock}, \citenamefont {Strasberg},\ and\ \citenamefont {Modi}}]{Dowling2023b}%
  \BibitemOpen
  \bibfield  {author} {\bibinfo {author} {\bibfnamefont {N.}~\bibnamefont {Dowling}}, \bibinfo {author} {\bibfnamefont {P.}~\bibnamefont {Figueroa-Romero}}, \bibinfo {author} {\bibfnamefont {F.~A.}\ \bibnamefont {Pollock}}, \bibinfo {author} {\bibfnamefont {P.}~\bibnamefont {Strasberg}},\ and\ \bibinfo {author} {\bibfnamefont {K.}~\bibnamefont {Modi}},\ }\bibfield  {title} {\bibinfo {title} {Equilibration of multitime quantum processes in finite time intervals},\ }\href {http://dx.doi.org/10.21468/SciPostPhysCore.6.2.043} {\bibfield  {journal} {\bibinfo  {journal} {SciPost Physics Core}\ }\textbf {\bibinfo {volume} {6}} (\bibinfo {year} {2023}{\natexlab{b}})}\BibitemShut {NoStop}%
\bibitem [{\citenamefont {Glauber}(1963)}]{Glauber1963}%
  \BibitemOpen
  \bibfield  {author} {\bibinfo {author} {\bibfnamefont {R.~J.}\ \bibnamefont {Glauber}},\ }\bibfield  {title} {\bibinfo {title} {The quantum theory of optical coherence},\ }\href {https://doi.org/10.1103/physrev.130.2529} {\bibfield  {journal} {\bibinfo  {journal} {Physical Review}\ }\textbf {\bibinfo {volume} {130}},\ \bibinfo {pages} {2529–2539} (\bibinfo {year} {1963})}\BibitemShut {NoStop}%
\bibitem [{\citenamefont {Lax}(1968)}]{Lax68}%
  \BibitemOpen
  \bibfield  {author} {\bibinfo {author} {\bibfnamefont {M.}~\bibnamefont {Lax}},\ }\bibfield  {title} {\bibinfo {title} {Quantum noise. {X}{I}. {M}ultitime correspondence between quantum and classical stochastic processes},\ }\href {http://dx.doi.org/10.1103/PhysRev.172.350} {\bibfield  {journal} {\bibinfo  {journal} {Physical Review}\ }\textbf {\bibinfo {volume} {172}},\ \bibinfo {pages} {350–361} (\bibinfo {year} {1968})}\BibitemShut {NoStop}%
\bibitem [{\citenamefont {Gardiner}\ and\ \citenamefont {Zoller}(2004)}]{GardinerZoller2004}%
  \BibitemOpen
  \bibfield  {author} {\bibinfo {author} {\bibfnamefont {C.~W.}\ \bibnamefont {Gardiner}}\ and\ \bibinfo {author} {\bibfnamefont {P.}~\bibnamefont {Zoller}},\ }\href {https://www.springer.com/gp/book/9783540225110} {\emph {\bibinfo {title} {Quantum Noise: A Handbook of Markovian and Non-Markovian Quantum Stochastic Methods with Applications to Quantum Optics}}},\ \bibinfo {edition} {2nd}\ ed.\ (\bibinfo  {publisher} {Springer},\ \bibinfo {address} {Berlin},\ \bibinfo {year} {2004})\BibitemShut {NoStop}%
\bibitem [{\citenamefont {Landi}\ \emph {et~al.}(2024)\citenamefont {Landi}, \citenamefont {Kewming}, \citenamefont {Mitchison},\ and\ \citenamefont {Potts}}]{Landi2024}%
  \BibitemOpen
  \bibfield  {author} {\bibinfo {author} {\bibfnamefont {G.~T.}\ \bibnamefont {Landi}}, \bibinfo {author} {\bibfnamefont {M.~J.}\ \bibnamefont {Kewming}}, \bibinfo {author} {\bibfnamefont {M.~T.}\ \bibnamefont {Mitchison}},\ and\ \bibinfo {author} {\bibfnamefont {P.~P.}\ \bibnamefont {Potts}},\ }\bibfield  {title} {\bibinfo {title} {Current fluctuations in open quantum systems: Bridging the gap between quantum continuous measurements and full counting statistics},\ }\href {http://dx.doi.org/10.1103/PRXQuantum.5.020201} {\bibfield  {journal} {\bibinfo  {journal} {PRX Quantum}\ }\textbf {\bibinfo {volume} {5}} (\bibinfo {year} {2024})}\BibitemShut {NoStop}%
\bibitem [{\citenamefont {Li}\ \emph {et~al.}(2018)\citenamefont {Li}, \citenamefont {Hall},\ and\ \citenamefont {Wiseman}}]{Li2018}%
  \BibitemOpen
  \bibfield  {author} {\bibinfo {author} {\bibfnamefont {L.}~\bibnamefont {Li}}, \bibinfo {author} {\bibfnamefont {M.~J.}\ \bibnamefont {Hall}},\ and\ \bibinfo {author} {\bibfnamefont {H.~M.}\ \bibnamefont {Wiseman}},\ }\bibfield  {title} {\bibinfo {title} {Concepts of quantum non-markovianity: A hierarchy},\ }\href {https://doi.org/10.1016/j.physrep.2018.07.001} {\bibfield  {journal} {\bibinfo  {journal} {Physics Reports}\ }\textbf {\bibinfo {volume} {759}},\ \bibinfo {pages} {1–51} (\bibinfo {year} {2018})}\BibitemShut {NoStop}%
\bibitem [{\citenamefont {Ford}\ and\ \citenamefont {O’Connell}(1996)}]{Ford1996}%
  \BibitemOpen
  \bibfield  {author} {\bibinfo {author} {\bibfnamefont {G.~W.}\ \bibnamefont {Ford}}\ and\ \bibinfo {author} {\bibfnamefont {R.~F.}\ \bibnamefont {O’Connell}},\ }\bibfield  {title} {\bibinfo {title} {There is no quantum regression theorem},\ }\href {http://dx.doi.org/10.1103/PhysRevLett.77.798} {\bibfield  {journal} {\bibinfo  {journal} {Physical Review Letters}\ }\textbf {\bibinfo {volume} {77}},\ \bibinfo {pages} {798–801} (\bibinfo {year} {1996})}\BibitemShut {NoStop}%
\bibitem [{\citenamefont {Alonso}\ and\ \citenamefont {de~Vega}(2005)}]{Alonso2005}%
  \BibitemOpen
  \bibfield  {author} {\bibinfo {author} {\bibfnamefont {D.}~\bibnamefont {Alonso}}\ and\ \bibinfo {author} {\bibfnamefont {I.}~\bibnamefont {de~Vega}},\ }\bibfield  {title} {\bibinfo {title} {Multiple-time correlation functions for non-markovian interaction: Beyond the quantum regression theorem},\ }\href {http://dx.doi.org/10.1103/PhysRevLett.94.200403} {\bibfield  {journal} {\bibinfo  {journal} {Physical Review Letters}\ }\textbf {\bibinfo {volume} {94}} (\bibinfo {year} {2005})}\BibitemShut {NoStop}%
\bibitem [{\citenamefont {Alonso}\ and\ \citenamefont {de~Vega}(2007)}]{Alonso2007}%
  \BibitemOpen
  \bibfield  {author} {\bibinfo {author} {\bibfnamefont {D.}~\bibnamefont {Alonso}}\ and\ \bibinfo {author} {\bibfnamefont {I.}~\bibnamefont {de~Vega}},\ }\bibfield  {title} {\bibinfo {title} {Hierarchy of equations of multiple-time correlation functions},\ }\href {http://dx.doi.org/10.1103/PhysRevA.75.052108} {\bibfield  {journal} {\bibinfo  {journal} {Physical Review A}\ }\textbf {\bibinfo {volume} {75}} (\bibinfo {year} {2007})}\BibitemShut {NoStop}%
\bibitem [{\citenamefont {Guarnieri}\ \emph {et~al.}(2014)\citenamefont {Guarnieri}, \citenamefont {Smirne},\ and\ \citenamefont {Vacchini}}]{Guarnieri2014}%
  \BibitemOpen
  \bibfield  {author} {\bibinfo {author} {\bibfnamefont {G.}~\bibnamefont {Guarnieri}}, \bibinfo {author} {\bibfnamefont {A.}~\bibnamefont {Smirne}},\ and\ \bibinfo {author} {\bibfnamefont {B.}~\bibnamefont {Vacchini}},\ }\bibfield  {title} {\bibinfo {title} {Quantum regression theorem and non-markovianity of quantum dynamics},\ }\href {http://dx.doi.org/10.1103/PhysRevA.90.022110} {\bibfield  {journal} {\bibinfo  {journal} {Physical Review A}\ }\textbf {\bibinfo {volume} {90}} (\bibinfo {year} {2014})}\BibitemShut {NoStop}%
\bibitem [{\citenamefont {de~Vega}\ and\ \citenamefont {Alonso}(2017)}]{VA17}%
  \BibitemOpen
  \bibfield  {author} {\bibinfo {author} {\bibfnamefont {I.}~\bibnamefont {de~Vega}}\ and\ \bibinfo {author} {\bibfnamefont {D.}~\bibnamefont {Alonso}},\ }\bibfield  {title} {\bibinfo {title} {Dynamics of non-markovian open quantum systems},\ }\href {https://doi.org/10.1103/revmodphys.89.015001} {\bibfield  {journal} {\bibinfo  {journal} {Reviews of Modern Physics}\ }\textbf {\bibinfo {volume} {89}} (\bibinfo {year} {2017})}\BibitemShut {NoStop}%
\bibitem [{\citenamefont {Cosacchi}\ \emph {et~al.}(2021)\citenamefont {Cosacchi}, \citenamefont {Seidelmann}, \citenamefont {Cygorek}, \citenamefont {Vagov}, \citenamefont {Reiter},\ and\ \citenamefont {Axt}}]{Cosacchi2021}%
  \BibitemOpen
  \bibfield  {author} {\bibinfo {author} {\bibfnamefont {M.}~\bibnamefont {Cosacchi}}, \bibinfo {author} {\bibfnamefont {T.}~\bibnamefont {Seidelmann}}, \bibinfo {author} {\bibfnamefont {M.}~\bibnamefont {Cygorek}}, \bibinfo {author} {\bibfnamefont {A.}~\bibnamefont {Vagov}}, \bibinfo {author} {\bibfnamefont {D.~E.}\ \bibnamefont {Reiter}},\ and\ \bibinfo {author} {\bibfnamefont {V.~M.}\ \bibnamefont {Axt}},\ }\bibfield  {title} {\bibinfo {title} {Accuracy of the quantum regression theorem for photon emission from a quantum dot},\ }\href {http://dx.doi.org/10.1103/PhysRevLett.127.100402} {\bibfield  {journal} {\bibinfo  {journal} {Physical Review Letters}\ }\textbf {\bibinfo {volume} {127}} (\bibinfo {year} {2021})}\BibitemShut {NoStop}%
\bibitem [{\citenamefont {Kurt}(2021)}]{Kurt2021}%
  \BibitemOpen
  \bibfield  {author} {\bibinfo {author} {\bibfnamefont {A.}~\bibnamefont {Kurt}},\ }\bibfield  {title} {\bibinfo {title} {Two-time correlation functions beyond quantum regression theorem: effect of external noise},\ }\href {http://dx.doi.org/10.1007/s11128-021-03153-6} {\bibfield  {journal} {\bibinfo  {journal} {Quantum Information Processing}\ }\textbf {\bibinfo {volume} {20}} (\bibinfo {year} {2021})}\BibitemShut {NoStop}%
\bibitem [{\citenamefont {Lonigro}\ and\ \citenamefont {Chruściński}(2022{\natexlab{a}})}]{Lonigro2022a}%
  \BibitemOpen
  \bibfield  {author} {\bibinfo {author} {\bibfnamefont {D.}~\bibnamefont {Lonigro}}\ and\ \bibinfo {author} {\bibfnamefont {D.}~\bibnamefont {Chruściński}},\ }\bibfield  {title} {\bibinfo {title} {Quantum regression in dephasing phenomena},\ }\href {http://dx.doi.org/10.1088/1751-8121/ac6a2d} {\bibfield  {journal} {\bibinfo  {journal} {Journal of Physics A: Mathematical and Theoretical}\ }\textbf {\bibinfo {volume} {55}},\ \bibinfo {pages} {225308} (\bibinfo {year} {2022}{\natexlab{a}})}\BibitemShut {NoStop}%
\bibitem [{\citenamefont {Lonigro}\ and\ \citenamefont {Chruściński}(2022{\natexlab{b}})}]{Lonigro2022b}%
  \BibitemOpen
  \bibfield  {author} {\bibinfo {author} {\bibfnamefont {D.}~\bibnamefont {Lonigro}}\ and\ \bibinfo {author} {\bibfnamefont {D.}~\bibnamefont {Chruściński}},\ }\bibfield  {title} {\bibinfo {title} {Quantum regression beyond the born-markov approximation for generalized spin-boson models},\ }\href {http://dx.doi.org/10.1103/PhysRevA.105.052435} {\bibfield  {journal} {\bibinfo  {journal} {Physical Review A}\ }\textbf {\bibinfo {volume} {105}} (\bibinfo {year} {2022}{\natexlab{b}})}\BibitemShut {NoStop}%
\bibitem [{\citenamefont {Giarmatzi}\ and\ \citenamefont {Costa}(2021)}]{GC21}%
  \BibitemOpen
  \bibfield  {author} {\bibinfo {author} {\bibfnamefont {C.}~\bibnamefont {Giarmatzi}}\ and\ \bibinfo {author} {\bibfnamefont {F.}~\bibnamefont {Costa}},\ }\bibfield  {title} {\bibinfo {title} {Witnessing quantum memory in non-markovian processes},\ }\href {https://doi.org/10.22331/q-2021-04-26-440} {\bibfield  {journal} {\bibinfo  {journal} {Quantum}\ }\textbf {\bibinfo {volume} {5}},\ \bibinfo {pages} {440} (\bibinfo {year} {2021})}\BibitemShut {NoStop}%
\bibitem [{\citenamefont {Milz}\ \emph {et~al.}(2021)\citenamefont {Milz}, \citenamefont {Spee}, \citenamefont {Xu}, \citenamefont {Pollock}, \citenamefont {Modi},\ and\ \citenamefont {Gühne}}]{MSXPMG21}%
  \BibitemOpen
  \bibfield  {author} {\bibinfo {author} {\bibfnamefont {S.}~\bibnamefont {Milz}}, \bibinfo {author} {\bibfnamefont {C.}~\bibnamefont {Spee}}, \bibinfo {author} {\bibfnamefont {Z.-P.}\ \bibnamefont {Xu}}, \bibinfo {author} {\bibfnamefont {F.}~\bibnamefont {Pollock}}, \bibinfo {author} {\bibfnamefont {K.}~\bibnamefont {Modi}},\ and\ \bibinfo {author} {\bibfnamefont {O.}~\bibnamefont {Gühne}},\ }\bibfield  {title} {\bibinfo {title} {Genuine multipartite entanglement in time},\ }\href {https://doi.org/10.21468/scipostphys.10.6.141} {\bibfield  {journal} {\bibinfo  {journal} {{SciPost} Physics}\ }\textbf {\bibinfo {volume} {10}} (\bibinfo {year} {2021})}\BibitemShut {NoStop}%
\bibitem [{\citenamefont {Nery}\ \emph {et~al.}(2021)\citenamefont {Nery}, \citenamefont {Quintino}, \citenamefont {Gu{\'{e}}rin}, \citenamefont {Maciel},\ and\ \citenamefont {Vianna}}]{NQAGMV21}%
  \BibitemOpen
  \bibfield  {author} {\bibinfo {author} {\bibfnamefont {M.}~\bibnamefont {Nery}}, \bibinfo {author} {\bibfnamefont {M.~T.}\ \bibnamefont {Quintino}}, \bibinfo {author} {\bibfnamefont {P.~A.}\ \bibnamefont {Gu{\'{e}}rin}}, \bibinfo {author} {\bibfnamefont {T.~O.}\ \bibnamefont {Maciel}},\ and\ \bibinfo {author} {\bibfnamefont {R.~O.}\ \bibnamefont {Vianna}},\ }\bibfield  {title} {\bibinfo {title} {Simple and maximally robust processes with no classical common-cause or direct-cause explanation},\ }\href {https://doi.org/10.22331/q-2021-09-09-538} {\bibfield  {journal} {\bibinfo  {journal} {Quantum}\ }\textbf {\bibinfo {volume} {5}},\ \bibinfo {pages} {538} (\bibinfo {year} {2021})}\BibitemShut {NoStop}%
\bibitem [{\citenamefont {Taranto}\ \emph {et~al.}(2024)\citenamefont {Taranto}, \citenamefont {Quintino}, \citenamefont {Murao},\ and\ \citenamefont {Milz}}]{Taranto2024}%
  \BibitemOpen
  \bibfield  {author} {\bibinfo {author} {\bibfnamefont {P.}~\bibnamefont {Taranto}}, \bibinfo {author} {\bibfnamefont {M.~T.}\ \bibnamefont {Quintino}}, \bibinfo {author} {\bibfnamefont {M.}~\bibnamefont {Murao}},\ and\ \bibinfo {author} {\bibfnamefont {S.}~\bibnamefont {Milz}},\ }\bibfield  {title} {\bibinfo {title} {Characterising the hierarchy of multi-time quantum processes with classical memory},\ }\href {https://doi.org/10.22331/q-2024-05-02-1328} {\bibfield  {journal} {\bibinfo  {journal} {Quantum}\ }\textbf {\bibinfo {volume} {8}},\ \bibinfo {pages} {1328} (\bibinfo {year} {2024})}\BibitemShut {NoStop}%
\bibitem [{\citenamefont {Roy}\ \emph {et~al.}(2024)\citenamefont {Roy}, \citenamefont {Srivastava}, \citenamefont {Mahanti}, \citenamefont {Giarmatzi},\ and\ \citenamefont {Gilchrist}}]{Roy2024}%
  \BibitemOpen
  \bibfield  {author} {\bibinfo {author} {\bibfnamefont {A.~K.}\ \bibnamefont {Roy}}, \bibinfo {author} {\bibfnamefont {V.}~\bibnamefont {Srivastava}}, \bibinfo {author} {\bibfnamefont {S.}~\bibnamefont {Mahanti}}, \bibinfo {author} {\bibfnamefont {C.}~\bibnamefont {Giarmatzi}},\ and\ \bibinfo {author} {\bibfnamefont {A.}~\bibnamefont {Gilchrist}},\ }\bibfield  {title} {\bibinfo {title} {Semi-device-independent certification of quantum non-markovianity using sequential random access codes},\ }\href {http://dx.doi.org/10.1103/PhysRevA.110.012608} {\bibfield  {journal} {\bibinfo  {journal} {Physical Review A}\ }\textbf {\bibinfo {volume} {110}} (\bibinfo {year} {2024})}\BibitemShut {NoStop}%
\bibitem [{\citenamefont {Bäcker}\ \emph {et~al.}(2024)\citenamefont {Bäcker}, \citenamefont {Beyer},\ and\ \citenamefont {Strunz}}]{BBS24}%
  \BibitemOpen
  \bibfield  {author} {\bibinfo {author} {\bibfnamefont {C.}~\bibnamefont {Bäcker}}, \bibinfo {author} {\bibfnamefont {K.}~\bibnamefont {Beyer}},\ and\ \bibinfo {author} {\bibfnamefont {W.~T.}\ \bibnamefont {Strunz}},\ }\bibfield  {title} {\bibinfo {title} {Local disclosure of quantum memory in non-markovian dynamics},\ }\href {http://dx.doi.org/10.1103/PhysRevLett.132.060402} {\bibfield  {journal} {\bibinfo  {journal} {Physical Review Letters}\ }\textbf {\bibinfo {volume} {132}} (\bibinfo {year} {2024})}\BibitemShut {NoStop}%
\bibitem [{\citenamefont {Santos}\ \emph {et~al.}(2024)\citenamefont {Santos}, \citenamefont {Xu}, \citenamefont {Piilo},\ and\ \citenamefont {Gühne}}]{Santos2024}%
  \BibitemOpen
  \bibfield  {author} {\bibinfo {author} {\bibfnamefont {L.}~\bibnamefont {Santos}}, \bibinfo {author} {\bibfnamefont {Z.-P.}\ \bibnamefont {Xu}}, \bibinfo {author} {\bibfnamefont {J.}~\bibnamefont {Piilo}},\ and\ \bibinfo {author} {\bibfnamefont {O.}~\bibnamefont {Gühne}},\ }\bibfield  {title} {\bibinfo {title} {Quantifying information flow in quantum processes},\ }\href@noop {} {\  (\bibinfo {year} {2024})},\ \Eprint {https://arxiv.org/abs/2402.04213v2} {arXiv:2402.04213v2 [quant-ph]} \BibitemShut {NoStop}%
\bibitem [{\citenamefont {Ohst}\ \emph {et~al.}(2024)\citenamefont {Ohst}, \citenamefont {Zhang}, \citenamefont {Nguyen}, \citenamefont {Plávala},\ and\ \citenamefont {Quintino}}]{Ohst2024}%
  \BibitemOpen
  \bibfield  {author} {\bibinfo {author} {\bibfnamefont {T.-A.}\ \bibnamefont {Ohst}}, \bibinfo {author} {\bibfnamefont {S.}~\bibnamefont {Zhang}}, \bibinfo {author} {\bibfnamefont {H.~C.}\ \bibnamefont {Nguyen}}, \bibinfo {author} {\bibfnamefont {M.}~\bibnamefont {Plávala}},\ and\ \bibinfo {author} {\bibfnamefont {M.~T.}\ \bibnamefont {Quintino}},\ }\bibfield  {title} {\bibinfo {title} {Characterising memory in quantum channel discrimination via constrained separability problems},\ }\href@noop {} {\  (\bibinfo {year} {2024})},\ \Eprint {https://arxiv.org/abs/2411.08110v1} {arXiv:2411.08110v1 [quant-ph]} \BibitemShut {NoStop}%
\bibitem [{\citenamefont {Goswami}\ \emph {et~al.}(2024)\citenamefont {Goswami}, \citenamefont {Roy}, \citenamefont {Srivastava}, \citenamefont {Perez}, \citenamefont {Giarmatzi}, \citenamefont {Gilchrist},\ and\ \citenamefont {Costa}}]{Goswami2024}%
  \BibitemOpen
  \bibfield  {author} {\bibinfo {author} {\bibfnamefont {K.}~\bibnamefont {Goswami}}, \bibinfo {author} {\bibfnamefont {A.~K.}\ \bibnamefont {Roy}}, \bibinfo {author} {\bibfnamefont {V.}~\bibnamefont {Srivastava}}, \bibinfo {author} {\bibfnamefont {B.}~\bibnamefont {Perez}}, \bibinfo {author} {\bibfnamefont {C.}~\bibnamefont {Giarmatzi}}, \bibinfo {author} {\bibfnamefont {A.}~\bibnamefont {Gilchrist}},\ and\ \bibinfo {author} {\bibfnamefont {F.}~\bibnamefont {Costa}},\ }\bibfield  {title} {\bibinfo {title} {Hamiltonian characterisation of multi-time processes with classical memory},\ }\href@noop {} {\  (\bibinfo {year} {2024})},\ \Eprint {https://arxiv.org/abs/2412.01998v1} {arXiv:2412.01998v1 [quant-ph]} \BibitemShut {NoStop}%
\bibitem [{\citenamefont {Foligno}\ \emph {et~al.}(2023)\citenamefont {Foligno}, \citenamefont {Zhou},\ and\ \citenamefont {Bertini}}]{Foligno2023}%
  \BibitemOpen
  \bibfield  {author} {\bibinfo {author} {\bibfnamefont {A.}~\bibnamefont {Foligno}}, \bibinfo {author} {\bibfnamefont {T.}~\bibnamefont {Zhou}},\ and\ \bibinfo {author} {\bibfnamefont {B.}~\bibnamefont {Bertini}},\ }\bibfield  {title} {\bibinfo {title} {Temporal entanglement in chaotic quantum circuits},\ }\href {http://dx.doi.org/10.1103/PhysRevX.13.041008} {\bibfield  {journal} {\bibinfo  {journal} {Physical Review X}\ }\textbf {\bibinfo {volume} {13}} (\bibinfo {year} {2023})}\BibitemShut {NoStop}%
\bibitem [{\citenamefont {Carignano}\ \emph {et~al.}(2024)\citenamefont {Carignano}, \citenamefont {Marimón},\ and\ \citenamefont {Tagliacozzo}}]{Carignano2024}%
  \BibitemOpen
  \bibfield  {author} {\bibinfo {author} {\bibfnamefont {S.}~\bibnamefont {Carignano}}, \bibinfo {author} {\bibfnamefont {C.~R.}\ \bibnamefont {Marimón}},\ and\ \bibinfo {author} {\bibfnamefont {L.}~\bibnamefont {Tagliacozzo}},\ }\bibfield  {title} {\bibinfo {title} {Temporal entropy and the complexity of computing the expectation value of local operators after a quench},\ }\href {http://dx.doi.org/10.1103/PhysRevResearch.6.033021} {\bibfield  {journal} {\bibinfo  {journal} {Physical Review Research}\ }\textbf {\bibinfo {volume} {6}} (\bibinfo {year} {2024})}\BibitemShut {NoStop}%
\bibitem [{\citenamefont {Yu}\ \emph {et~al.}(2025)\citenamefont {Yu}, \citenamefont {Ohst}, \citenamefont {Nguyen},\ and\ \citenamefont {Nimmrichter}}]{Yu2025}%
  \BibitemOpen
  \bibfield  {author} {\bibinfo {author} {\bibfnamefont {M.}~\bibnamefont {Yu}}, \bibinfo {author} {\bibfnamefont {T.-A.}\ \bibnamefont {Ohst}}, \bibinfo {author} {\bibfnamefont {H.-C.}\ \bibnamefont {Nguyen}},\ and\ \bibinfo {author} {\bibfnamefont {S.}~\bibnamefont {Nimmrichter}},\ }\bibfield  {title} {\bibinfo {title} {Quantum memory in spontaneous emission processes},\ }\Eprint {https://arxiv.org/abs/2504.08605v1} {arXiv:2504.08605v1 [quant-ph]}  (\bibinfo {year} {2025})\BibitemShut {NoStop}%
\bibitem [{\citenamefont {Lindblad}(1976)}]{Lindblad76}%
  \BibitemOpen
  \bibfield  {author} {\bibinfo {author} {\bibfnamefont {G.}~\bibnamefont {Lindblad}},\ }\bibfield  {title} {\bibinfo {title} {On the generators of quantum dynamical semigroups},\ }\href {https://doi.org/10.1007/bf01608499} {\bibfield  {journal} {\bibinfo  {journal} {Communications in Mathematical Physics}\ }\textbf {\bibinfo {volume} {48}},\ \bibinfo {pages} {119} (\bibinfo {year} {1976})}\BibitemShut {NoStop}%
\bibitem [{\citenamefont {Gorini}\ \emph {et~al.}(1976)\citenamefont {Gorini}, \citenamefont {Kossakowski},\ and\ \citenamefont {Sudarshan}}]{GKS76}%
  \BibitemOpen
  \bibfield  {author} {\bibinfo {author} {\bibfnamefont {V.}~\bibnamefont {Gorini}}, \bibinfo {author} {\bibfnamefont {A.}~\bibnamefont {Kossakowski}},\ and\ \bibinfo {author} {\bibfnamefont {E.~C.~G.}\ \bibnamefont {Sudarshan}},\ }\bibfield  {title} {\bibinfo {title} {Completely positive dynamical semigroups of n-level systems},\ }\href {https://doi.org/10.1063/1.522979} {\bibfield  {journal} {\bibinfo  {journal} {Journal of Mathematical Physics}\ }\textbf {\bibinfo {volume} {17}},\ \bibinfo {pages} {821} (\bibinfo {year} {1976})}\BibitemShut {NoStop}%
\bibitem [{\citenamefont {Albash}\ \emph {et~al.}(2012)\citenamefont {Albash}, \citenamefont {Boixo}, \citenamefont {Lidar},\ and\ \citenamefont {Zanardi}}]{Albash2012}%
  \BibitemOpen
  \bibfield  {author} {\bibinfo {author} {\bibfnamefont {T.}~\bibnamefont {Albash}}, \bibinfo {author} {\bibfnamefont {S.}~\bibnamefont {Boixo}}, \bibinfo {author} {\bibfnamefont {D.~A.}\ \bibnamefont {Lidar}},\ and\ \bibinfo {author} {\bibfnamefont {P.}~\bibnamefont {Zanardi}},\ }\bibfield  {title} {\bibinfo {title} {Quantum adiabatic markovian master equations},\ }\href {http://dx.doi.org/10.1088/1367-2630/14/12/123016} {\bibfield  {journal} {\bibinfo  {journal} {New Journal of Physics}\ }\textbf {\bibinfo {volume} {14}},\ \bibinfo {pages} {123016} (\bibinfo {year} {2012})}\BibitemShut {NoStop}%
\bibitem [{\citenamefont {Yamaguchi}\ \emph {et~al.}(2017)\citenamefont {Yamaguchi}, \citenamefont {Yuge},\ and\ \citenamefont {Ogawa}}]{Yamaguchi2017}%
  \BibitemOpen
  \bibfield  {author} {\bibinfo {author} {\bibfnamefont {M.}~\bibnamefont {Yamaguchi}}, \bibinfo {author} {\bibfnamefont {T.}~\bibnamefont {Yuge}},\ and\ \bibinfo {author} {\bibfnamefont {T.}~\bibnamefont {Ogawa}},\ }\bibfield  {title} {\bibinfo {title} {Markovian quantum master equation beyond adiabatic regime},\ }\href {http://dx.doi.org/10.1103/PhysRevE.95.012136} {\bibfield  {journal} {\bibinfo  {journal} {Physical Review E}\ }\textbf {\bibinfo {volume} {95}} (\bibinfo {year} {2017})}\BibitemShut {NoStop}%
\bibitem [{\citenamefont {Dann}\ \emph {et~al.}(2018)\citenamefont {Dann}, \citenamefont {Levy},\ and\ \citenamefont {Kosloff}}]{Dann2018}%
  \BibitemOpen
  \bibfield  {author} {\bibinfo {author} {\bibfnamefont {R.}~\bibnamefont {Dann}}, \bibinfo {author} {\bibfnamefont {A.}~\bibnamefont {Levy}},\ and\ \bibinfo {author} {\bibfnamefont {R.}~\bibnamefont {Kosloff}},\ }\bibfield  {title} {\bibinfo {title} {Time-dependent markovian quantum master equation},\ }\href {http://dx.doi.org/10.1103/PhysRevA.98.052129} {\bibfield  {journal} {\bibinfo  {journal} {Physical Review A}\ }\textbf {\bibinfo {volume} {98}} (\bibinfo {year} {2018})}\BibitemShut {NoStop}%
\bibitem [{\citenamefont {Di~Meglio}\ \emph {et~al.}(2024)\citenamefont {Di~Meglio}, \citenamefont {Plenio},\ and\ \citenamefont {Huelga}}]{DiMeglio2024}%
  \BibitemOpen
  \bibfield  {author} {\bibinfo {author} {\bibfnamefont {G.}~\bibnamefont {Di~Meglio}}, \bibinfo {author} {\bibfnamefont {M.~B.}\ \bibnamefont {Plenio}},\ and\ \bibinfo {author} {\bibfnamefont {S.~F.}\ \bibnamefont {Huelga}},\ }\bibfield  {title} {\bibinfo {title} {Time dependent markovian master equation beyond the adiabatic limit},\ }\href {http://dx.doi.org/10.22331/q-2024-11-21-1534} {\bibfield  {journal} {\bibinfo  {journal} {Quantum}\ }\textbf {\bibinfo {volume} {8}},\ \bibinfo {pages} {1534} (\bibinfo {year} {2024})}\BibitemShut {NoStop}%
\bibitem [{\citenamefont {Rivas}\ \emph {et~al.}(2014)\citenamefont {Rivas}, \citenamefont {Huelga},\ and\ \citenamefont {Plenio}}]{RHP14}%
  \BibitemOpen
  \bibfield  {author} {\bibinfo {author} {\bibfnamefont {{\'{A}}.}~\bibnamefont {Rivas}}, \bibinfo {author} {\bibfnamefont {S.~F.}\ \bibnamefont {Huelga}},\ and\ \bibinfo {author} {\bibfnamefont {M.~B.}\ \bibnamefont {Plenio}},\ }\bibfield  {title} {\bibinfo {title} {Quantum non-markovianity: characterization, quantification and detection},\ }\href {https://doi.org/10.1088/0034-4885/77/9/094001} {\bibfield  {journal} {\bibinfo  {journal} {Reports on Progress in Physics}\ }\textbf {\bibinfo {volume} {77}},\ \bibinfo {pages} {094001} (\bibinfo {year} {2014})}\BibitemShut {NoStop}%
\bibitem [{\citenamefont {Pollock}\ \emph {et~al.}(2018{\natexlab{a}})\citenamefont {Pollock}, \citenamefont {Rodríguez-Rosario}, \citenamefont {Frauenheim}, \citenamefont {Paternostro},\ and\ \citenamefont {Modi}}]{Pollock2018a}%
  \BibitemOpen
  \bibfield  {author} {\bibinfo {author} {\bibfnamefont {F.~A.}\ \bibnamefont {Pollock}}, \bibinfo {author} {\bibfnamefont {C.}~\bibnamefont {Rodríguez-Rosario}}, \bibinfo {author} {\bibfnamefont {T.}~\bibnamefont {Frauenheim}}, \bibinfo {author} {\bibfnamefont {M.}~\bibnamefont {Paternostro}},\ and\ \bibinfo {author} {\bibfnamefont {K.}~\bibnamefont {Modi}},\ }\bibfield  {title} {\bibinfo {title} {Operational markov condition for quantum processes},\ }\href {http://dx.doi.org/10.1103/PhysRevLett.120.040405} {\bibfield  {journal} {\bibinfo  {journal} {Physical Review Letters}\ }\textbf {\bibinfo {volume} {120}} (\bibinfo {year} {2018}{\natexlab{a}})}\BibitemShut {NoStop}%
\bibitem [{\citenamefont {Pollock}\ \emph {et~al.}(2018{\natexlab{b}})\citenamefont {Pollock}, \citenamefont {Rodríguez-Rosario}, \citenamefont {Frauenheim}, \citenamefont {Paternostro},\ and\ \citenamefont {Modi}}]{Pollock2018b}%
  \BibitemOpen
  \bibfield  {author} {\bibinfo {author} {\bibfnamefont {F.~A.}\ \bibnamefont {Pollock}}, \bibinfo {author} {\bibfnamefont {C.}~\bibnamefont {Rodríguez-Rosario}}, \bibinfo {author} {\bibfnamefont {T.}~\bibnamefont {Frauenheim}}, \bibinfo {author} {\bibfnamefont {M.}~\bibnamefont {Paternostro}},\ and\ \bibinfo {author} {\bibfnamefont {K.}~\bibnamefont {Modi}},\ }\bibfield  {title} {\bibinfo {title} {Non-markovian quantum processes: Complete framework and efficient characterization},\ }\href {http://dx.doi.org/10.1103/PhysRevA.97.012127} {\bibfield  {journal} {\bibinfo  {journal} {Physical Review A}\ }\textbf {\bibinfo {volume} {97}} (\bibinfo {year} {2018}{\natexlab{b}})}\BibitemShut {NoStop}%
\bibitem [{\citenamefont {Milz}\ and\ \citenamefont {Modi}(2021)}]{MM21}%
  \BibitemOpen
  \bibfield  {author} {\bibinfo {author} {\bibfnamefont {S.}~\bibnamefont {Milz}}\ and\ \bibinfo {author} {\bibfnamefont {K.}~\bibnamefont {Modi}},\ }\bibfield  {title} {\bibinfo {title} {Quantum stochastic processes and quantum non-markovian phenomena},\ }\href {https://doi.org/10.1103/prxquantum.2.030201} {\bibfield  {journal} {\bibinfo  {journal} {{PRX} Quantum}\ }\textbf {\bibinfo {volume} {2}} (\bibinfo {year} {2021})}\BibitemShut {NoStop}%
\bibitem [{\citenamefont {Choi}(1975)}]{Choi75}%
  \BibitemOpen
  \bibfield  {author} {\bibinfo {author} {\bibfnamefont {M.-D.}\ \bibnamefont {Choi}},\ }\bibfield  {title} {\bibinfo {title} {Completely positive linear maps on complex matrices},\ }\href {https://doi.org/10.1016/0024-3795\%2875\%2990075-0} {\bibfield  {journal} {\bibinfo  {journal} {Linear Algebra and its Applications}\ }\textbf {\bibinfo {volume} {10}},\ \bibinfo {pages} {285} (\bibinfo {year} {1975})}\BibitemShut {NoStop}%
\bibitem [{\citenamefont {Jamio{\l}kowski}(1972)}]{J72}%
  \BibitemOpen
  \bibfield  {author} {\bibinfo {author} {\bibfnamefont {A.}~\bibnamefont {Jamio{\l}kowski}},\ }\bibfield  {title} {\bibinfo {title} {Linear transformations which preserve trace and positive semidefiniteness of operators},\ }\href {https://doi.org/10.1016/0034-4877\%2872\%2990011-0} {\bibfield  {journal} {\bibinfo  {journal} {Reports on Mathematical Physics}\ }\textbf {\bibinfo {volume} {3}},\ \bibinfo {pages} {275} (\bibinfo {year} {1972})}\BibitemShut {NoStop}%
\bibitem [{\citenamefont {White}\ \emph {et~al.}(2022)\citenamefont {White}, \citenamefont {Pollock}, \citenamefont {Hollenberg}, \citenamefont {Modi},\ and\ \citenamefont {Hill}}]{White2022}%
  \BibitemOpen
  \bibfield  {author} {\bibinfo {author} {\bibfnamefont {G.}~\bibnamefont {White}}, \bibinfo {author} {\bibfnamefont {F.}~\bibnamefont {Pollock}}, \bibinfo {author} {\bibfnamefont {L.}~\bibnamefont {Hollenberg}}, \bibinfo {author} {\bibfnamefont {K.}~\bibnamefont {Modi}},\ and\ \bibinfo {author} {\bibfnamefont {C.}~\bibnamefont {Hill}},\ }\bibfield  {title} {\bibinfo {title} {Non-markovian quantum process tomography},\ }\href {http://dx.doi.org/10.1103/PRXQuantum.3.020344} {\bibfield  {journal} {\bibinfo  {journal} {PRX Quantum}\ }\textbf {\bibinfo {volume} {3}} (\bibinfo {year} {2022})}\BibitemShut {NoStop}%
\bibitem [{\citenamefont {Chiribella}\ \emph {et~al.}(2008{\natexlab{a}})\citenamefont {Chiribella}, \citenamefont {D{\textquotesingle}Ariano},\ and\ \citenamefont {Perinotti}}]{CDP08a}%
  \BibitemOpen
  \bibfield  {author} {\bibinfo {author} {\bibfnamefont {G.}~\bibnamefont {Chiribella}}, \bibinfo {author} {\bibfnamefont {G.~M.}\ \bibnamefont {D{\textquotesingle}Ariano}},\ and\ \bibinfo {author} {\bibfnamefont {P.}~\bibnamefont {Perinotti}},\ }\bibfield  {title} {\bibinfo {title} {Transforming quantum operations: Quantum supermaps},\ }\href {https://doi.org/10.1209/0295-5075/83/30004} {\bibfield  {journal} {\bibinfo  {journal} {{EPL} (Europhysics Letters)}\ }\textbf {\bibinfo {volume} {83}},\ \bibinfo {pages} {30004} (\bibinfo {year} {2008}{\natexlab{a}})}\BibitemShut {NoStop}%
\bibitem [{\citenamefont {Chiribella}\ \emph {et~al.}(2008{\natexlab{b}})\citenamefont {Chiribella}, \citenamefont {D'Ariano},\ and\ \citenamefont {Perinotti}}]{CDP08b}%
  \BibitemOpen
  \bibfield  {author} {\bibinfo {author} {\bibfnamefont {G.}~\bibnamefont {Chiribella}}, \bibinfo {author} {\bibfnamefont {G.~M.}\ \bibnamefont {D'Ariano}},\ and\ \bibinfo {author} {\bibfnamefont {P.}~\bibnamefont {Perinotti}},\ }\bibfield  {title} {\bibinfo {title} {Quantum circuit architecture},\ }\href {https://doi.org/10.1103/physrevlett.101.060401} {\bibfield  {journal} {\bibinfo  {journal} {Physical Review Letters}\ }\textbf {\bibinfo {volume} {101}} (\bibinfo {year} {2008}{\natexlab{b}})}\BibitemShut {NoStop}%
\bibitem [{\citenamefont {Chiribella}\ \emph {et~al.}(2009)\citenamefont {Chiribella}, \citenamefont {D'Ariano},\ and\ \citenamefont {Perinotti}}]{CDP09}%
  \BibitemOpen
  \bibfield  {author} {\bibinfo {author} {\bibfnamefont {G.}~\bibnamefont {Chiribella}}, \bibinfo {author} {\bibfnamefont {G.~M.}\ \bibnamefont {D'Ariano}},\ and\ \bibinfo {author} {\bibfnamefont {P.}~\bibnamefont {Perinotti}},\ }\bibfield  {title} {\bibinfo {title} {Theoretical framework for quantum networks},\ }\href {https://doi.org/10.1103/physreva.80.022339} {\bibfield  {journal} {\bibinfo  {journal} {Physical Review A}\ }\textbf {\bibinfo {volume} {80}} (\bibinfo {year} {2009})}\BibitemShut {NoStop}%
\bibitem [{\citenamefont {Yadin}\ \emph {et~al.}(2016)\citenamefont {Yadin}, \citenamefont {Ma}, \citenamefont {Girolami}, \citenamefont {Gu},\ and\ \citenamefont {Vedral}}]{Yadin2016}%
  \BibitemOpen
  \bibfield  {author} {\bibinfo {author} {\bibfnamefont {B.}~\bibnamefont {Yadin}}, \bibinfo {author} {\bibfnamefont {J.}~\bibnamefont {Ma}}, \bibinfo {author} {\bibfnamefont {D.}~\bibnamefont {Girolami}}, \bibinfo {author} {\bibfnamefont {M.}~\bibnamefont {Gu}},\ and\ \bibinfo {author} {\bibfnamefont {V.}~\bibnamefont {Vedral}},\ }\bibfield  {title} {\bibinfo {title} {Quantum processes which do not use coherence},\ }\href {http://dx.doi.org/10.1103/PhysRevX.6.041028} {\bibfield  {journal} {\bibinfo  {journal} {Physical Review X}\ }\textbf {\bibinfo {volume} {6}} (\bibinfo {year} {2016})}\BibitemShut {NoStop}%
\bibitem [{Note1()}]{Note1}%
  \BibitemOpen
  \bibinfo {note} {This is because any strictly incoherent operation, which includes pure dephasing, can be performed in terms of incoherent unitary operations (see Ref.~\cite {Yadin2016} for details).}\BibitemShut {Stop}%
\bibitem [{\citenamefont {Liu}\ \emph {et~al.}(2011)\citenamefont {Liu}, \citenamefont {Li}, \citenamefont {Huang}, \citenamefont {Li}, \citenamefont {Guo}, \citenamefont {Laine}, \citenamefont {Breuer},\ and\ \citenamefont {Piilo}}]{LLHLGLBP11}%
  \BibitemOpen
  \bibfield  {author} {\bibinfo {author} {\bibfnamefont {B.-H.}\ \bibnamefont {Liu}}, \bibinfo {author} {\bibfnamefont {L.}~\bibnamefont {Li}}, \bibinfo {author} {\bibfnamefont {Y.-F.}\ \bibnamefont {Huang}}, \bibinfo {author} {\bibfnamefont {C.-F.}\ \bibnamefont {Li}}, \bibinfo {author} {\bibfnamefont {G.-C.}\ \bibnamefont {Guo}}, \bibinfo {author} {\bibfnamefont {E.-M.}\ \bibnamefont {Laine}}, \bibinfo {author} {\bibfnamefont {H.-P.}\ \bibnamefont {Breuer}},\ and\ \bibinfo {author} {\bibfnamefont {J.}~\bibnamefont {Piilo}},\ }\bibfield  {title} {\bibinfo {title} {Experimental control of the transition from markovian to non-markovian dynamics of open~quantum~systems},\ }\href {https://doi.org/10.1038/nphys2085} {\bibfield  {journal} {\bibinfo  {journal} {Nature Physics}\ }\textbf {\bibinfo {volume} {7}},\ \bibinfo {pages} {931} (\bibinfo {year} {2011})}\BibitemShut {NoStop}%
\bibitem [{\citenamefont {Liu}\ \emph {et~al.}(2018)\citenamefont {Liu}, \citenamefont {Lyyra}, \citenamefont {Sun}, \citenamefont {Liu}, \citenamefont {Li}, \citenamefont {Guo}, \citenamefont {Maniscalco},\ and\ \citenamefont {Piilo}}]{LLSLLGMP18}%
  \BibitemOpen
  \bibfield  {author} {\bibinfo {author} {\bibfnamefont {Z.-D.}\ \bibnamefont {Liu}}, \bibinfo {author} {\bibfnamefont {H.}~\bibnamefont {Lyyra}}, \bibinfo {author} {\bibfnamefont {Y.-N.}\ \bibnamefont {Sun}}, \bibinfo {author} {\bibfnamefont {B.-H.}\ \bibnamefont {Liu}}, \bibinfo {author} {\bibfnamefont {C.-F.}\ \bibnamefont {Li}}, \bibinfo {author} {\bibfnamefont {G.-C.}\ \bibnamefont {Guo}}, \bibinfo {author} {\bibfnamefont {S.}~\bibnamefont {Maniscalco}},\ and\ \bibinfo {author} {\bibfnamefont {J.}~\bibnamefont {Piilo}},\ }\bibfield  {title} {\bibinfo {title} {Experimental implementation of fully controlled dephasing dynamics and synthetic spectral densities},\ }\href {https://doi.org/10.1038/s41467-018-05817-x} {\bibfield  {journal} {\bibinfo  {journal} {Nature Communications}\ }\textbf {\bibinfo {volume} {9}} (\bibinfo {year} {2018})}\BibitemShut {NoStop}%
\bibitem [{Note2()}]{Note2}%
  \BibitemOpen
  \bibinfo {note} {It is therefore consistent with previous definitions of ``quantum memory'' in different contexts; cf. \cite {Nielsen2003,Chiribella2008,Lvovsky2009,Terhal2015,Rosset2018}.}\BibitemShut {Stop}%
\bibitem [{\citenamefont {Gühne}\ and\ \citenamefont {Tóth}(2009)}]{Guehne2009}%
  \BibitemOpen
  \bibfield  {author} {\bibinfo {author} {\bibfnamefont {O.}~\bibnamefont {Gühne}}\ and\ \bibinfo {author} {\bibfnamefont {G.}~\bibnamefont {Tóth}},\ }\bibfield  {title} {\bibinfo {title} {Entanglement detection},\ }\href {http://dx.doi.org/10.1016/j.physrep.2009.02.004} {\bibfield  {journal} {\bibinfo  {journal} {Physics Reports}\ }\textbf {\bibinfo {volume} {474}},\ \bibinfo {pages} {1–75} (\bibinfo {year} {2009})}\BibitemShut {NoStop}%
\bibitem [{\citenamefont {Skrzypczyk}\ and\ \citenamefont {Cavalcanti}(2023)}]{SC2023}%
  \BibitemOpen
  \bibfield  {author} {\bibinfo {author} {\bibfnamefont {P.}~\bibnamefont {Skrzypczyk}}\ and\ \bibinfo {author} {\bibfnamefont {D.}~\bibnamefont {Cavalcanti}},\ }\href {http://dx.doi.org/10.1088/978-0-7503-3343-6} {\emph {\bibinfo {title} {Semidefinite Programming in Quantum Information Science}}}\ (\bibinfo  {publisher} {IOP Publishing},\ \bibinfo {year} {2023})\BibitemShut {NoStop}%
\bibitem [{\citenamefont {Berta}\ \emph {et~al.}(2016)\citenamefont {Berta}, \citenamefont {Fawzi},\ and\ \citenamefont {Scholz}}]{Berta2016}%
  \BibitemOpen
  \bibfield  {author} {\bibinfo {author} {\bibfnamefont {M.}~\bibnamefont {Berta}}, \bibinfo {author} {\bibfnamefont {O.}~\bibnamefont {Fawzi}},\ and\ \bibinfo {author} {\bibfnamefont {V.~B.}\ \bibnamefont {Scholz}},\ }\bibfield  {title} {\bibinfo {title} {Quantum bilinear optimization},\ }\href {http://dx.doi.org/10.1137/15M1037731} {\bibfield  {journal} {\bibinfo  {journal} {SIAM Journal on Optimization}\ }\textbf {\bibinfo {volume} {26}},\ \bibinfo {pages} {1529–1564} (\bibinfo {year} {2016})}\BibitemShut {NoStop}%
\bibitem [{\citenamefont {Berta}\ \emph {et~al.}(2021)\citenamefont {Berta}, \citenamefont {Borderi}, \citenamefont {Fawzi},\ and\ \citenamefont {Scholz}}]{Berta2021}%
  \BibitemOpen
  \bibfield  {author} {\bibinfo {author} {\bibfnamefont {M.}~\bibnamefont {Berta}}, \bibinfo {author} {\bibfnamefont {F.}~\bibnamefont {Borderi}}, \bibinfo {author} {\bibfnamefont {O.}~\bibnamefont {Fawzi}},\ and\ \bibinfo {author} {\bibfnamefont {V.~B.}\ \bibnamefont {Scholz}},\ }\bibfield  {title} {\bibinfo {title} {Semidefinite programming hierarchies for constrained bilinear optimization},\ }\href {http://dx.doi.org/10.1007/s10107-021-01650-1} {\bibfield  {journal} {\bibinfo  {journal} {Mathematical Programming}\ }\textbf {\bibinfo {volume} {194}},\ \bibinfo {pages} {781–829} (\bibinfo {year} {2021})}\BibitemShut {NoStop}%
\bibitem [{Note3()}]{Note3}%
  \BibitemOpen
  \bibinfo {note} {The link product between two operators, say $X$ and $Y$ on Hilbert spaces $\protect \mathcal {H}_{\protect \rm A}\otimes \protect \mathcal {H}_{\protect \rm B}$ and $\protect \mathcal {H}_{\protect \rm B}\otimes \protect \mathcal {H}_{\protect \rm C}$ respectively, is defined as $$X\star Y=\protect \mathrm {Tr}\protect \hspace {0.05cm}_{\protect \rm B}\protect \big [\protect \big (X^{\Gamma _{\protect \rm B}}\otimes \protect \mathrm {id}_{\protect \rm C}\protect \big )\protect \big (\protect \mathrm {id}_{\protect \rm A}\otimes Y\protect \big )\protect \big ].$$ It can be understood as a simple way to translate composition between channels in the CJ picture; see Ref.~\cite {CDP09} for details.}\BibitemShut {Stop}%
\bibitem [{\citenamefont {Terhal}\ and\ \citenamefont {Horodecki}(2000)}]{Terhal2000}%
  \BibitemOpen
  \bibfield  {author} {\bibinfo {author} {\bibfnamefont {B.~M.}\ \bibnamefont {Terhal}}\ and\ \bibinfo {author} {\bibfnamefont {P.}~\bibnamefont {Horodecki}},\ }\bibfield  {title} {\bibinfo {title} {Schmidt number for density matrices},\ }\href {http://dx.doi.org/10.1103/PhysRevA.61.040301} {\bibfield  {journal} {\bibinfo  {journal} {Physical Review A}\ }\textbf {\bibinfo {volume} {61}} (\bibinfo {year} {2000})}\BibitemShut {NoStop}%
\bibitem [{\citenamefont {Moroder}\ and\ \citenamefont {Gittsovich}(2012)}]{Moroder2012}%
  \BibitemOpen
  \bibfield  {author} {\bibinfo {author} {\bibfnamefont {T.}~\bibnamefont {Moroder}}\ and\ \bibinfo {author} {\bibfnamefont {O.}~\bibnamefont {Gittsovich}},\ }\bibfield  {title} {\bibinfo {title} {Calibration-robust entanglement detection beyond bell inequalities},\ }\bibfield  {journal} {\bibinfo  {journal} {Physical Review A}\ }\textbf {\bibinfo {volume} {85}},\ \href {https://doi.org/10.1103/physreva.85.032301} {10.1103/physreva.85.032301} (\bibinfo {year} {2012})\BibitemShut {NoStop}%
\bibitem [{\citenamefont {Breuer}\ and\ \citenamefont {Petruccione}(2007)}]{BP07}%
  \BibitemOpen
  \bibfield  {author} {\bibinfo {author} {\bibfnamefont {H.-P.}\ \bibnamefont {Breuer}}\ and\ \bibinfo {author} {\bibfnamefont {F.}~\bibnamefont {Petruccione}},\ }\href {https://doi.org/10.1093/acprof\%3Aoso/9780199213900.001.0001} {\emph {\bibinfo {title} {The Theory of Open Quantum Systems}}}\ (\bibinfo  {publisher} {Oxford University {Press, Oxford}},\ \bibinfo {year} {2007})\BibitemShut {NoStop}%
\bibitem [{\citenamefont {Brune}\ \emph {et~al.}(1996)\citenamefont {Brune}, \citenamefont {Schmidt-Kaler}, \citenamefont {Maali}, \citenamefont {Dreyer}, \citenamefont {Hagley}, \citenamefont {Raimond},\ and\ \citenamefont {Haroche}}]{Brune1996}%
  \BibitemOpen
  \bibfield  {author} {\bibinfo {author} {\bibfnamefont {M.}~\bibnamefont {Brune}}, \bibinfo {author} {\bibfnamefont {F.}~\bibnamefont {Schmidt-Kaler}}, \bibinfo {author} {\bibfnamefont {A.}~\bibnamefont {Maali}}, \bibinfo {author} {\bibfnamefont {J.}~\bibnamefont {Dreyer}}, \bibinfo {author} {\bibfnamefont {E.}~\bibnamefont {Hagley}}, \bibinfo {author} {\bibfnamefont {J.~M.}\ \bibnamefont {Raimond}},\ and\ \bibinfo {author} {\bibfnamefont {S.}~\bibnamefont {Haroche}},\ }\bibfield  {title} {\bibinfo {title} {Quantum rabi oscillation: A direct test of field quantization in a cavity},\ }\href {http://dx.doi.org/10.1103/PhysRevLett.76.1800} {\bibfield  {journal} {\bibinfo  {journal} {Physical Review Letters}\ }\textbf {\bibinfo {volume} {76}},\ \bibinfo {pages} {1800–1803} (\bibinfo {year} {1996})}\BibitemShut {NoStop}%
\bibitem [{\citenamefont {Budroni}\ \emph {et~al.}(2013)\citenamefont {Budroni}, \citenamefont {Moroder}, \citenamefont {Kleinmann},\ and\ \citenamefont {Gühne}}]{Budroni2013}%
  \BibitemOpen
  \bibfield  {author} {\bibinfo {author} {\bibfnamefont {C.}~\bibnamefont {Budroni}}, \bibinfo {author} {\bibfnamefont {T.}~\bibnamefont {Moroder}}, \bibinfo {author} {\bibfnamefont {M.}~\bibnamefont {Kleinmann}},\ and\ \bibinfo {author} {\bibfnamefont {O.}~\bibnamefont {Gühne}},\ }\bibfield  {title} {\bibinfo {title} {Bounding temporal quantum correlations},\ }\href {http://dx.doi.org/10.1103/PhysRevLett.111.020403} {\bibfield  {journal} {\bibinfo  {journal} {Physical Review Letters}\ }\textbf {\bibinfo {volume} {111}} (\bibinfo {year} {2013})}\BibitemShut {NoStop}%
\bibitem [{\citenamefont {Budroni}\ and\ \citenamefont {Emary}(2014)}]{Budroni2014}%
  \BibitemOpen
  \bibfield  {author} {\bibinfo {author} {\bibfnamefont {C.}~\bibnamefont {Budroni}}\ and\ \bibinfo {author} {\bibfnamefont {C.}~\bibnamefont {Emary}},\ }\bibfield  {title} {\bibinfo {title} {Temporal quantum correlations and leggett-garg inequalities in multilevel systems},\ }\href {http://dx.doi.org/10.1103/PhysRevLett.113.050401} {\bibfield  {journal} {\bibinfo  {journal} {Physical Review Letters}\ }\textbf {\bibinfo {volume} {113}} (\bibinfo {year} {2014})}\BibitemShut {NoStop}%
\bibitem [{\citenamefont {Budroni}(2016)}]{Budroni2016}%
  \BibitemOpen
  \bibfield  {author} {\bibinfo {author} {\bibfnamefont {C.}~\bibnamefont {Budroni}},\ }\href {http://dx.doi.org/10.1007/978-3-319-24169-2} {\emph {\bibinfo {title} {Temporal Quantum Correlations and Hidden Variable Models}}}\ (\bibinfo  {publisher} {Springer International Publishing},\ \bibinfo {year} {2016})\BibitemShut {NoStop}%
\bibitem [{\citenamefont {Hoffmann}\ \emph {et~al.}(2018)\citenamefont {Hoffmann}, \citenamefont {Spee}, \citenamefont {Gühne},\ and\ \citenamefont {Budroni}}]{Hoffmann2018}%
  \BibitemOpen
  \bibfield  {author} {\bibinfo {author} {\bibfnamefont {J.}~\bibnamefont {Hoffmann}}, \bibinfo {author} {\bibfnamefont {C.}~\bibnamefont {Spee}}, \bibinfo {author} {\bibfnamefont {O.}~\bibnamefont {Gühne}},\ and\ \bibinfo {author} {\bibfnamefont {C.}~\bibnamefont {Budroni}},\ }\bibfield  {title} {\bibinfo {title} {Structure of temporal correlations of a qubit},\ }\href {http://dx.doi.org/10.1088/1367-2630/aae87f} {\bibfield  {journal} {\bibinfo  {journal} {New Journal of Physics}\ }\textbf {\bibinfo {volume} {20}},\ \bibinfo {pages} {102001} (\bibinfo {year} {2018})}\BibitemShut {NoStop}%
\bibitem [{\citenamefont {Costa}\ \emph {et~al.}(2018)\citenamefont {Costa}, \citenamefont {Ringbauer}, \citenamefont {Goggin}, \citenamefont {White},\ and\ \citenamefont {Fedrizzi}}]{Costa2018}%
  \BibitemOpen
  \bibfield  {author} {\bibinfo {author} {\bibfnamefont {F.}~\bibnamefont {Costa}}, \bibinfo {author} {\bibfnamefont {M.}~\bibnamefont {Ringbauer}}, \bibinfo {author} {\bibfnamefont {M.~E.}\ \bibnamefont {Goggin}}, \bibinfo {author} {\bibfnamefont {A.~G.}\ \bibnamefont {White}},\ and\ \bibinfo {author} {\bibfnamefont {A.}~\bibnamefont {Fedrizzi}},\ }\bibfield  {title} {\bibinfo {title} {Unifying framework for spatial and temporal quantum correlations},\ }\href {http://dx.doi.org/10.1103/PhysRevA.98.012328} {\bibfield  {journal} {\bibinfo  {journal} {Physical Review A}\ }\textbf {\bibinfo {volume} {98}} (\bibinfo {year} {2018})}\BibitemShut {NoStop}%
\bibitem [{\citenamefont {Spee}\ \emph {et~al.}(2020)\citenamefont {Spee}, \citenamefont {Siebeneich}, \citenamefont {Gloger}, \citenamefont {Kaufmann}, \citenamefont {Johanning}, \citenamefont {Kleinmann}, \citenamefont {Wunderlich},\ and\ \citenamefont {Gühne}}]{Spee2020}%
  \BibitemOpen
  \bibfield  {author} {\bibinfo {author} {\bibfnamefont {C.}~\bibnamefont {Spee}}, \bibinfo {author} {\bibfnamefont {H.}~\bibnamefont {Siebeneich}}, \bibinfo {author} {\bibfnamefont {T.~F.}\ \bibnamefont {Gloger}}, \bibinfo {author} {\bibfnamefont {P.}~\bibnamefont {Kaufmann}}, \bibinfo {author} {\bibfnamefont {M.}~\bibnamefont {Johanning}}, \bibinfo {author} {\bibfnamefont {M.}~\bibnamefont {Kleinmann}}, \bibinfo {author} {\bibfnamefont {C.}~\bibnamefont {Wunderlich}},\ and\ \bibinfo {author} {\bibfnamefont {O.}~\bibnamefont {Gühne}},\ }\bibfield  {title} {\bibinfo {title} {Genuine temporal correlations can certify the quantum dimension},\ }\href {https://doi.org/10.1088/1367-2630/ab6d42} {\bibfield  {journal} {\bibinfo  {journal} {New Journal of Physics}\ }\textbf {\bibinfo {volume} {22}},\ \bibinfo {pages} {023028} (\bibinfo {year} {2020})}\BibitemShut {NoStop}%
\bibitem [{\citenamefont {Budroni}\ \emph {et~al.}(2021)\citenamefont {Budroni}, \citenamefont {Vitagliano},\ and\ \citenamefont {Woods}}]{Budroni2021}%
  \BibitemOpen
  \bibfield  {author} {\bibinfo {author} {\bibfnamefont {C.}~\bibnamefont {Budroni}}, \bibinfo {author} {\bibfnamefont {G.}~\bibnamefont {Vitagliano}},\ and\ \bibinfo {author} {\bibfnamefont {M.~P.}\ \bibnamefont {Woods}},\ }\bibfield  {title} {\bibinfo {title} {Ticking-clock performance enhanced by nonclassical temporal correlations},\ }\href {http://dx.doi.org/10.1103/PhysRevResearch.3.033051} {\bibfield  {journal} {\bibinfo  {journal} {Physical Review Research}\ }\textbf {\bibinfo {volume} {3}} (\bibinfo {year} {2021})}\BibitemShut {NoStop}%
\bibitem [{\citenamefont {Vieira}\ \emph {et~al.}(2024)\citenamefont {Vieira}, \citenamefont {Milz}, \citenamefont {Vitagliano},\ and\ \citenamefont {Budroni}}]{Vieira2024}%
  \BibitemOpen
  \bibfield  {author} {\bibinfo {author} {\bibfnamefont {L.~B.}\ \bibnamefont {Vieira}}, \bibinfo {author} {\bibfnamefont {S.}~\bibnamefont {Milz}}, \bibinfo {author} {\bibfnamefont {G.}~\bibnamefont {Vitagliano}},\ and\ \bibinfo {author} {\bibfnamefont {C.}~\bibnamefont {Budroni}},\ }\bibfield  {title} {\bibinfo {title} {Witnessing environment dimension through temporal correlations},\ }\href {http://dx.doi.org/10.22331/q-2024-01-10-1224} {\bibfield  {journal} {\bibinfo  {journal} {Quantum}\ }\textbf {\bibinfo {volume} {8}},\ \bibinfo {pages} {1224} (\bibinfo {year} {2024})}\BibitemShut {NoStop}%
\bibitem [{\citenamefont {Strathearn}\ \emph {et~al.}(2018)\citenamefont {Strathearn}, \citenamefont {Kirton}, \citenamefont {Kilda}, \citenamefont {Keeling},\ and\ \citenamefont {Lovett}}]{Strathearn2018}%
  \BibitemOpen
  \bibfield  {author} {\bibinfo {author} {\bibfnamefont {A.}~\bibnamefont {Strathearn}}, \bibinfo {author} {\bibfnamefont {P.}~\bibnamefont {Kirton}}, \bibinfo {author} {\bibfnamefont {D.}~\bibnamefont {Kilda}}, \bibinfo {author} {\bibfnamefont {J.}~\bibnamefont {Keeling}},\ and\ \bibinfo {author} {\bibfnamefont {B.~W.}\ \bibnamefont {Lovett}},\ }\bibfield  {title} {\bibinfo {title} {Efficient non-markovian quantum dynamics using time-evolving matrix product operators},\ }\href {http://dx.doi.org/10.1038/s41467-018-05617-3} {\bibfield  {journal} {\bibinfo  {journal} {Nature Communications}\ }\textbf {\bibinfo {volume} {9}} (\bibinfo {year} {2018})}\BibitemShut {NoStop}%
\bibitem [{\citenamefont {Binder}\ \emph {et~al.}(2018)\citenamefont {Binder}, \citenamefont {Thompson},\ and\ \citenamefont {Gu}}]{Binder2018}%
  \BibitemOpen
  \bibfield  {author} {\bibinfo {author} {\bibfnamefont {F.~C.}\ \bibnamefont {Binder}}, \bibinfo {author} {\bibfnamefont {J.}~\bibnamefont {Thompson}},\ and\ \bibinfo {author} {\bibfnamefont {M.}~\bibnamefont {Gu}},\ }\bibfield  {title} {\bibinfo {title} {Practical unitary simulator for non-markovian complex processes},\ }\href {http://dx.doi.org/10.1103/PhysRevLett.120.240502} {\bibfield  {journal} {\bibinfo  {journal} {Physical Review Letters}\ }\textbf {\bibinfo {volume} {120}} (\bibinfo {year} {2018})}\BibitemShut {NoStop}%
\bibitem [{\citenamefont {Jørgensen}\ and\ \citenamefont {Pollock}(2019)}]{J2019}%
  \BibitemOpen
  \bibfield  {author} {\bibinfo {author} {\bibfnamefont {M.~R.}\ \bibnamefont {Jørgensen}}\ and\ \bibinfo {author} {\bibfnamefont {F.~A.}\ \bibnamefont {Pollock}},\ }\bibfield  {title} {\bibinfo {title} {Exploiting the causal tensor network structure of quantum processes to efficiently simulate non-markovian path integrals},\ }\href {http://dx.doi.org/10.1103/PhysRevLett.123.240602} {\bibfield  {journal} {\bibinfo  {journal} {Physical Review Letters}\ }\textbf {\bibinfo {volume} {123}} (\bibinfo {year} {2019})}\BibitemShut {NoStop}%
\bibitem [{\citenamefont {Cygorek}\ \emph {et~al.}(2022)\citenamefont {Cygorek}, \citenamefont {Cosacchi}, \citenamefont {Vagov}, \citenamefont {Axt}, \citenamefont {Lovett}, \citenamefont {Keeling},\ and\ \citenamefont {Gauger}}]{Cygorek2022}%
  \BibitemOpen
  \bibfield  {author} {\bibinfo {author} {\bibfnamefont {M.}~\bibnamefont {Cygorek}}, \bibinfo {author} {\bibfnamefont {M.}~\bibnamefont {Cosacchi}}, \bibinfo {author} {\bibfnamefont {A.}~\bibnamefont {Vagov}}, \bibinfo {author} {\bibfnamefont {V.~M.}\ \bibnamefont {Axt}}, \bibinfo {author} {\bibfnamefont {B.~W.}\ \bibnamefont {Lovett}}, \bibinfo {author} {\bibfnamefont {J.}~\bibnamefont {Keeling}},\ and\ \bibinfo {author} {\bibfnamefont {E.~M.}\ \bibnamefont {Gauger}},\ }\bibfield  {title} {\bibinfo {title} {Simulation of open quantum systems by automated compression of arbitrary environments},\ }\href {http://dx.doi.org/10.1038/s41567-022-01544-9} {\bibfield  {journal} {\bibinfo  {journal} {Nature Physics}\ }\textbf {\bibinfo {volume} {18}},\ \bibinfo {pages} {662–668} (\bibinfo {year} {2022})}\BibitemShut {NoStop}%
\bibitem [{\citenamefont {White}\ \emph {et~al.}(2023)\citenamefont {White}, \citenamefont {Jurcevic}, \citenamefont {Hill},\ and\ \citenamefont {Modi}}]{White2023}%
  \BibitemOpen
  \bibfield  {author} {\bibinfo {author} {\bibfnamefont {G.~A.~L.}\ \bibnamefont {White}}, \bibinfo {author} {\bibfnamefont {P.}~\bibnamefont {Jurcevic}}, \bibinfo {author} {\bibfnamefont {C.~D.}\ \bibnamefont {Hill}},\ and\ \bibinfo {author} {\bibfnamefont {K.}~\bibnamefont {Modi}},\ }\bibfield  {title} {\bibinfo {title} {Unifying non-markovian characterisation with an efficient and self-consistent framework},\ }\href@noop {} {\  (\bibinfo {year} {2023})},\ \Eprint {https://arxiv.org/abs/2312.08454v1} {arXiv:2312.08454v1 [quant-ph]} \BibitemShut {NoStop}%
\bibitem [{\citenamefont {Dowling}\ \emph {et~al.}(2024)\citenamefont {Dowling}, \citenamefont {Modi}, \citenamefont {Muñoz}, \citenamefont {Singh},\ and\ \citenamefont {White}}]{Dowling2024}%
  \BibitemOpen
  \bibfield  {author} {\bibinfo {author} {\bibfnamefont {N.}~\bibnamefont {Dowling}}, \bibinfo {author} {\bibfnamefont {K.}~\bibnamefont {Modi}}, \bibinfo {author} {\bibfnamefont {R.~N.}\ \bibnamefont {Muñoz}}, \bibinfo {author} {\bibfnamefont {S.}~\bibnamefont {Singh}},\ and\ \bibinfo {author} {\bibfnamefont {G.~A.~L.}\ \bibnamefont {White}},\ }\bibfield  {title} {\bibinfo {title} {Capturing long-range memory structures with tree-geometry process tensors},\ }\href {http://dx.doi.org/10.1103/PhysRevX.14.041018} {\bibfield  {journal} {\bibinfo  {journal} {Physical Review X}\ }\textbf {\bibinfo {volume} {14}} (\bibinfo {year} {2024})}\BibitemShut {NoStop}%
\bibitem [{\citenamefont {Cygorek}\ and\ \citenamefont {Gauger}(2024)}]{Cygorek2024a}%
  \BibitemOpen
  \bibfield  {author} {\bibinfo {author} {\bibfnamefont {M.}~\bibnamefont {Cygorek}}\ and\ \bibinfo {author} {\bibfnamefont {E.~M.}\ \bibnamefont {Gauger}},\ }\bibfield  {title} {\bibinfo {title} {Understanding and utilizing the inner bonds of process tensors},\ }\href@noop {} {\  (\bibinfo {year} {2024})},\ \Eprint {https://arxiv.org/abs/2404.01287v1} {arXiv:2404.01287v1 [quant-ph]} \BibitemShut {NoStop}%
\bibitem [{\citenamefont {Cygorek}\ \emph {et~al.}(2024)\citenamefont {Cygorek}, \citenamefont {Lovett}, \citenamefont {Keeling},\ and\ \citenamefont {Gauger}}]{Cygorek2024b}%
  \BibitemOpen
  \bibfield  {author} {\bibinfo {author} {\bibfnamefont {M.}~\bibnamefont {Cygorek}}, \bibinfo {author} {\bibfnamefont {B.~W.}\ \bibnamefont {Lovett}}, \bibinfo {author} {\bibfnamefont {J.}~\bibnamefont {Keeling}},\ and\ \bibinfo {author} {\bibfnamefont {E.~M.}\ \bibnamefont {Gauger}},\ }\bibfield  {title} {\bibinfo {title} {Tree-like process tensor contraction for automated compression of environments},\ }\href@noop {} {\  (\bibinfo {year} {2024})},\ \Eprint {https://arxiv.org/abs/2405.16548v1} {arXiv:2405.16548v1 [quant-ph]} \BibitemShut {NoStop}%
\bibitem [{\citenamefont {Li}\ \emph {et~al.}(2025)\citenamefont {Li}, \citenamefont {He}, \citenamefont {Zheng}, \citenamefont {Dong}, \citenamefont {Luan}, \citenamefont {Yu},\ and\ \citenamefont {Zhang}}]{Li2025}%
  \BibitemOpen
  \bibfield  {author} {\bibinfo {author} {\bibfnamefont {Z.-T.}\ \bibnamefont {Li}}, \bibinfo {author} {\bibfnamefont {X.-L.}\ \bibnamefont {He}}, \bibinfo {author} {\bibfnamefont {C.-C.}\ \bibnamefont {Zheng}}, \bibinfo {author} {\bibfnamefont {Y.-Q.}\ \bibnamefont {Dong}}, \bibinfo {author} {\bibfnamefont {T.}~\bibnamefont {Luan}}, \bibinfo {author} {\bibfnamefont {X.-T.}\ \bibnamefont {Yu}},\ and\ \bibinfo {author} {\bibfnamefont {Z.-C.}\ \bibnamefont {Zhang}},\ }\bibfield  {title} {\bibinfo {title} {Quantum comb tomography via learning isometries on {S}tiefel manifold},\ }\href {http://dx.doi.org/10.1103/PhysRevLett.134.010803} {\bibfield  {journal} {\bibinfo  {journal} {Physical Review Letters}\ }\textbf {\bibinfo {volume} {134}} (\bibinfo {year} {2025})}\BibitemShut {NoStop}%
\bibitem [{\citenamefont {Nielsen}(2003)}]{Nielsen2003}%
  \BibitemOpen
  \bibfield  {author} {\bibinfo {author} {\bibfnamefont {M.~A.}\ \bibnamefont {Nielsen}},\ }\bibfield  {title} {\bibinfo {title} {Quantum computation by measurement and quantum memory},\ }\href {http://dx.doi.org/10.1016/S0375-9601(02)01803-0} {\bibfield  {journal} {\bibinfo  {journal} {Physics Letters A}\ }\textbf {\bibinfo {volume} {308}},\ \bibinfo {pages} {96–100} (\bibinfo {year} {2003})}\BibitemShut {NoStop}%
\bibitem [{\citenamefont {Chiribella}\ \emph {et~al.}(2008{\natexlab{c}})\citenamefont {Chiribella}, \citenamefont {D’Ariano},\ and\ \citenamefont {Perinotti}}]{Chiribella2008}%
  \BibitemOpen
  \bibfield  {author} {\bibinfo {author} {\bibfnamefont {G.}~\bibnamefont {Chiribella}}, \bibinfo {author} {\bibfnamefont {G.~M.}\ \bibnamefont {D’Ariano}},\ and\ \bibinfo {author} {\bibfnamefont {P.}~\bibnamefont {Perinotti}},\ }\bibfield  {title} {\bibinfo {title} {Memory effects in quantum channel discrimination},\ }\href {http://dx.doi.org/10.1103/PhysRevLett.101.180501} {\bibfield  {journal} {\bibinfo  {journal} {Physical Review Letters}\ }\textbf {\bibinfo {volume} {101}} (\bibinfo {year} {2008}{\natexlab{c}})}\BibitemShut {NoStop}%
\bibitem [{\citenamefont {Lvovsky}\ \emph {et~al.}(2009)\citenamefont {Lvovsky}, \citenamefont {Sanders},\ and\ \citenamefont {Tittel}}]{Lvovsky2009}%
  \BibitemOpen
  \bibfield  {author} {\bibinfo {author} {\bibfnamefont {A.~I.}\ \bibnamefont {Lvovsky}}, \bibinfo {author} {\bibfnamefont {B.~C.}\ \bibnamefont {Sanders}},\ and\ \bibinfo {author} {\bibfnamefont {W.}~\bibnamefont {Tittel}},\ }\bibfield  {title} {\bibinfo {title} {Optical quantum memory},\ }\href {http://dx.doi.org/10.1038/nphoton.2009.231} {\bibfield  {journal} {\bibinfo  {journal} {Nature Photonics}\ }\textbf {\bibinfo {volume} {3}},\ \bibinfo {pages} {706–714} (\bibinfo {year} {2009})}\BibitemShut {NoStop}%
\bibitem [{\citenamefont {Terhal}(2015)}]{Terhal2015}%
  \BibitemOpen
  \bibfield  {author} {\bibinfo {author} {\bibfnamefont {B.~M.}\ \bibnamefont {Terhal}},\ }\bibfield  {title} {\bibinfo {title} {Quantum error correction for quantum memories},\ }\href {http://dx.doi.org/10.1103/RevModPhys.87.307} {\bibfield  {journal} {\bibinfo  {journal} {Reviews of Modern Physics}\ }\textbf {\bibinfo {volume} {87}},\ \bibinfo {pages} {307–346} (\bibinfo {year} {2015})}\BibitemShut {NoStop}%
\bibitem [{\citenamefont {Rosset}\ \emph {et~al.}(2018)\citenamefont {Rosset}, \citenamefont {Buscemi},\ and\ \citenamefont {Liang}}]{Rosset2018}%
  \BibitemOpen
  \bibfield  {author} {\bibinfo {author} {\bibfnamefont {D.}~\bibnamefont {Rosset}}, \bibinfo {author} {\bibfnamefont {F.}~\bibnamefont {Buscemi}},\ and\ \bibinfo {author} {\bibfnamefont {Y.-C.}\ \bibnamefont {Liang}},\ }\bibfield  {title} {\bibinfo {title} {Resource theory of quantum memories and their faithful verification with minimal assumptions},\ }\href {http://dx.doi.org/10.1103/PhysRevX.8.021033} {\bibfield  {journal} {\bibinfo  {journal} {Physical Review X}\ }\textbf {\bibinfo {volume} {8}} (\bibinfo {year} {2018})}\BibitemShut {NoStop}%
\bibitem [{\citenamefont {Oreshkov}\ \emph {et~al.}(2012)\citenamefont {Oreshkov}, \citenamefont {Costa},\ and\ \citenamefont {Brukner}}]{OCB12}%
  \BibitemOpen
  \bibfield  {author} {\bibinfo {author} {\bibfnamefont {O.}~\bibnamefont {Oreshkov}}, \bibinfo {author} {\bibfnamefont {F.}~\bibnamefont {Costa}},\ and\ \bibinfo {author} {\bibfnamefont {{\v{C}}.}~\bibnamefont {Brukner}},\ }\bibfield  {title} {\bibinfo {title} {Quantum correlations with no causal order},\ }\href {https://doi.org/10.1038/ncomms2076} {\bibfield  {journal} {\bibinfo  {journal} {Nature Communications}\ }\textbf {\bibinfo {volume} {3}} (\bibinfo {year} {2012})}\BibitemShut {NoStop}%
\bibitem [{\citenamefont {Ara{\'{u}}jo}\ \emph {et~al.}(2015)\citenamefont {Ara{\'{u}}jo}, \citenamefont {Branciard}, \citenamefont {Costa}, \citenamefont {Feix}, \citenamefont {Giarmatzi},\ and\ \citenamefont {Brukner}}]{ABCFGB15}%
  \BibitemOpen
  \bibfield  {author} {\bibinfo {author} {\bibfnamefont {M.}~\bibnamefont {Ara{\'{u}}jo}}, \bibinfo {author} {\bibfnamefont {C.}~\bibnamefont {Branciard}}, \bibinfo {author} {\bibfnamefont {F.}~\bibnamefont {Costa}}, \bibinfo {author} {\bibfnamefont {A.}~\bibnamefont {Feix}}, \bibinfo {author} {\bibfnamefont {C.}~\bibnamefont {Giarmatzi}},\ and\ \bibinfo {author} {\bibfnamefont {{\v{C}}.}~\bibnamefont {Brukner}},\ }\bibfield  {title} {\bibinfo {title} {Witnessing causal nonseparability},\ }\href {https://doi.org/10.1088/1367-2630/17/10/102001} {\bibfield  {journal} {\bibinfo  {journal} {New Journal of Physics}\ }\textbf {\bibinfo {volume} {17}},\ \bibinfo {pages} {102001} (\bibinfo {year} {2015})}\BibitemShut {NoStop}%
\bibitem [{Note4()}]{Note4}%
  \BibitemOpen
  \bibinfo {note} {To be fair, at this point one would only need to impose that $W$ be ``positive over tensor products'' to guarantee positivity of the probabilities (see Refs.~\cite {OCB12,ABCFGB15} for in-depth discussions). The need for $W$ to be \protect \emph {de facto} positive semidefinite comes from the possibility of probing the system $\protect \rm S$ with higher-order operations involving memory and entanglement.}\BibitemShut {Stop}%
\bibitem [{\citenamefont {Bengtsson}\ \emph {et~al.}(2024)\citenamefont {Bengtsson}, \citenamefont {Grassl},\ and\ \citenamefont {McConnell}}]{Bengtsson2024}%
  \BibitemOpen
  \bibfield  {author} {\bibinfo {author} {\bibfnamefont {I.}~\bibnamefont {Bengtsson}}, \bibinfo {author} {\bibfnamefont {M.}~\bibnamefont {Grassl}},\ and\ \bibinfo {author} {\bibfnamefont {G.}~\bibnamefont {McConnell}},\ }\bibfield  {title} {\bibinfo {title} {{SIC}-{POVM}s from {S}tark units: Dimensions $n^2+3=4p$, $p$ prime},\ }\href@noop {} {\  (\bibinfo {year} {2024})},\ \Eprint {https://arxiv.org/abs/2403.02872v1} {arXiv:2403.02872v1 [quant-ph]} \BibitemShut {NoStop}%
\bibitem [{\citenamefont {Appleby}\ \emph {et~al.}(2025)\citenamefont {Appleby}, \citenamefont {Flammia},\ and\ \citenamefont {Kopp}}]{Appleby2025}%
  \BibitemOpen
  \bibfield  {author} {\bibinfo {author} {\bibfnamefont {M.}~\bibnamefont {Appleby}}, \bibinfo {author} {\bibfnamefont {S.~T.}\ \bibnamefont {Flammia}},\ and\ \bibinfo {author} {\bibfnamefont {G.~S.}\ \bibnamefont {Kopp}},\ }\bibfield  {title} {\bibinfo {title} {A constructive approach to {Z}auner's conjecture via the {S}tark conjectures},\ }\href@noop {} {\  (\bibinfo {year} {2025})},\ \Eprint {https://arxiv.org/abs/2501.03970v1} {arXiv:2501.03970v1 [math.NT]} \BibitemShut {NoStop}%
\end{thebibliography}
\end{document}